\documentclass[runningheads]{llncs}
\usepackage[T1]{fontenc}
\usepackage{graphicx}
\usepackage{framed}
\usepackage{url}
\usepackage{algpseudocode}
\usepackage{algorithm}
\usepackage{cite}
\usepackage{todonotes}
\usepackage{comment}
\usepackage{amsmath,amssymb,amsfonts}

\usepackage{tcolorbox}
\tcbuselibrary{skins}
\usepackage{xcolor}

\colorlet{mix}{red!50!black}

\usepackage{hyperref}
\hypersetup{colorlinks={true},linkcolor={blue},citecolor=blue}

\usepackage{cleveref}

\newtheorem{observation}{Observation}

\usepackage{enumitem}
\usepackage{comment}
\usepackage{tikz}
\usepackage{epsfig}
\usepackage{wrapfig}
\usepackage{xspace}
\usetikzlibrary{shapes,backgrounds, patterns}

\colorlet{bscolor}{blue}

\newcommand{\mespshort}{{\sc ESP}\xspace}

\newcommand{\lesp}{\textsc{Eccentricity Shortest Path}} 
\newcommand{\esp}{\textsc{ESP}}

\newcommand{\mespfvs}{{\sc ESP}/${\sf fvs}+{\sf ecc}$\xspace}
\newcommand{\mespfappr}{{\sc ESP}/${\sf fvs}$\xspace}
\newcommand{\mespsvd}{{\sc ESP}/${\sf svd}$\xspace}
\newcommand{\mespdpd}{{\sc ESP}/${\sf dpd}$\xspace}
\newcommand{\mespdpdell}{{\sc ESP}/${\sf dpd}+{\sf ecc}$\xspace}
\newcommand{\mespcvd}{{\sc ESP}/${\sf cvd}$\xspace}
\newcommand{\dsfull}{{\sc Dominating Set}\xspace}
\newcommand{\ds}{{DS}\xspace}

\newcommand{\ty}{{\tilde{W}}\xspace}

\newcommand{\mespecc}{{\sc ESP}/$ecc$\xspace}

\newcommand{\annotatedmesp}{{\sc Annotated-ESP}\xspace}
\newcommand{\amesp}{{\sc Colorful Path-Cover}\xspace}

\newcommand{\disjointmesp}{{\sc Skeleton Testing}\xspace}
\newcommand{\probenskel}{{\sc  Ext-Skeleton Testing}\xspace}

\newcommand{\dmesp}{{\sc Skeleton Testing}\xspace}
\newcommand{\cP}{{\mathcal{P}}\xspace}
\newcommand{\OO}{{\mathcal{O}}\xspace}
\newcommand{\cF}{{\mathcal{F}}\xspace}
\newcommand{\cC}{{\mathcal{C}}\xspace}

\newcommand{\sch}{{\sf ch}\xspace}
\newcommand{\srem}{{\sf \eta}\xspace}
\newcommand{\ssign}{{\sf sg}\xspace}

\newcommand{\stype}{{\sf tp}\xspace}
\newcommand{\strue}{{\sf true}\xspace}
\newcommand{\sfalse}{{\sf false}\xspace}
\newcommand{\sP}{{(+)}\xspace}
\newcommand{\sM}{{(-)}\xspace}
\newcommand{\wh}{\textsf{W[2]}-hard}

\newcommand{\skel}{{\mathbb{S}}\xspace}
\newcommand{\enskel}{{\mathbb{ES}}\xspace}

\newcommand{\defparprob}[4]{
	\vspace{1mm}
	\noindent\fbox{
		\begin{minipage}{0.96\textwidth}
			\begin{tabular*}{\textwidth}{@{\extracolsep{\fill}}lr} #1  & {\bf{Parameter:}} #3 \\ \end{tabular*}
			{\bf{Input:}} #2  \\
			{\bf{Question:}} #4
		\end{minipage}
	}
	\vspace{1mm}
}

\newcommand{\defproblem}[3]{
		\vspace{2mm}
		\noindent\fbox{
			\begin{minipage}{0.96\textwidth}
				\begin{tabular*}{\textwidth}{@{\extracolsep{\fill}}lr} #1 \\ \end{tabular*}
				{\bf{Input:}} #2  \\
				{\bf{Question:}} #3
			\end{minipage}
		}
		\vspace{1mm}
	}

	\newcommand{\bscomment}[1]{\textcolor{bscolor}{BS:#1}}

	\newcommand{\Omit}[1]{}

	\newtheorem{reduction rule}{Reduction Rule}
	\usepackage{array}

	\usepackage{xcolor}
	\usepackage{comment}
	\usepackage{tcolorbox}
	\tcbuselibrary{breakable}
	\newtcolorbox{mybox}[2][]{colbacktitle=white,colback=white,coltitle=black,title={#2},fonttitle=\bfseries,#1, left = 2mm, right = 2mm, breakable}
	
\begin{document}
%
% \title{Contribution Title\thanks{Supported by organization x.}}
% %
% %\titlerunning{Abbreviated paper title}
% % If the paper title is too long for the running head, you can set
% % an abbreviated paper title here
% \maketitle  

	\title{Parameterized algorithms for Eccentricity Shortest Path Problem\thanks{A subset of the results of this paper were accepted to be presented at the 34th International Workshop on Combinatorial Algorithms (IWOCA 2023)} }
	%$k$-coloring in some graphclasses} %TODO Please add

\titlerunning{Parameterized algorithms for Eccentricity Shortest Path Problem} %TODO optional, please use if title is longer than one line

\author{Sriram Bhyravarapu\inst{1}\and 
%}{The Institute of Mathematical Sciences, HBNI, Chennai, India  }{sriramb@imsc.res.in}{}{}
Satyabrata Jana\inst{1}\and
%{The Institute of Mathematical Sciences, HBNI, Chennai, India }{satyamtma@gmail.com}{}{}
Lawqueen Kanesh\inst{2}\and 
%{Indian Institute of Technology, Jodhpur}{lawqueen@iitj.ac.in}{}{}
Saket Saurabh\inst{1,3}\and
%}{The Institute of Mathematical Sciences, HBNI, Chennai, India  \and University of Bergen, Norway }{saket@imsc.res.in}{}{}
Shaily Verma\inst{1}}

\institute{The Institute of Mathematical Sciences, HBNI, Chennai, India\\
\email{\{sriram, satyabrataj,saket,shailyverma\}@imsc.res.in}\\
\and Indian Institute of Technology Jodhpur, India\\
\email{lawqueen@iitj.ac.in}\\
\and University of Bergen, Norway}

\authorrunning{Bhyravarapu et al.} %TODO mandatory. First: Use abbreviated first/middle names. Second (only in severe cases): Use first author plus 'et al.'

%\Copyright{M. S. Ramanujan, Abhishek Sahu, Saket Saurabh, Shaily Verma} %TODO mandatory, please use full first names. LIPIcs license is "CC-BY";  http://creativecommons.org/licenses/by/3.0/

%\ccsdesc[100]{{Theory of computation $\rightarrow$ Fixed parameter tractability}} %TODO mandatory: Please choose ACM 2012 classifications from https://dl.acm.org/ccs/ccs_flat.cfm 

%TODO mandatory; please add comma-separated list of keywords

%\category{} %optional, e.g. invited paper

%\relatedversion{} 
%\acknowledgements{I want to thank \dots}%optional

%Editor-only macros:: begin (do not touch as author)\smallskip
%%%%%%%%%%%%%%%%%%%%%%%%%%%%%%%%%%
% \EventEditors{John Q. Open and Joan R. Access}
% \EventNoEds{2}
% \EventLongTitle{42nd Conference on Very Important Topics (CVIT 2016)}
% \EventShortTitle{CVIT 2016}
% \EventAcronym{CVIT}
% \EventYear{2016}
% \EventDate{December 24--27, 2016}
% \EventLocation{Little Whinging, United Kingdom}
% \EventLogo{}
% \SeriesVolume{42}
% \ArticleNo{23}
%%%%%%%%%%%%%%%%%%%%%%%%%%%%%%%%%%%%%%%%%%%%%%%%%%%%%%
\colorlet{bscolor}{blue}

%\input{refs}

%\DeclareUnicodeCharacter{B0}{\textless}
%\bibliographystyle{plainurl}
%\begin{document}
	\maketitle
	\begin{abstract}
		Given an undirected graph $ G=(V,E) $ and an integer $ \ell $, the  \lesp~(\esp) problem asks to 
check if 
%\todo[]{find or check for existence? }
there exists a shortest path $P$ such that for every vertex $v\in V(G)$, there is a vertex $w\in P$ such that $d_G(v,w)\leq \ell$, where $d_G(v,w)$ represents the distance between $v$ and $w$ in $G$. 
Dragan and Leitert [Theor. Comput. Sci. 2017]  studied the optimization version of this problem 
which asks to find the minimum $\ell$ for ESP and showed that it 
is {\sf NP}-hard even on planar bipartite graphs with maximum degree 3. 
They also showed that \esp~is {\sf W}[2]-hard when parameterized by $ \ell $. 
On the positive side, 
Ku\v cera and Such\'y [IWOCA 2021] showed that ESP is fixed-parameter tractable (FPT) when parameterized by 
modular width, cluster vertex deletion set, 
%distance to cluster, 
maximum leaf number, 
or the combined parameters 
disjoint paths deletion set and $ \ell $. 
It was asked as an 
open question in the same paper, 
if \esp~is FPT parameterized by disjoint paths deletion set 
or feedback vertex set.  
We answer these questions and obtain the 
following results: % on the \esp{} problem:
%\todo[]{removed the word partially}
\begin{enumerate}
	\item \esp{} is FPT 
	when 
	parameterized by disjoint paths deletion set, split vertex deletion set, 
	or the combined parameters feedback vertex set and $\ell$. 
 %(eccentricity of the graph). 
    %the eccentricity of the graph (the minimum eccentricity of any shortest path in the graph). 
	%the eccentricity of the graph. 
	
	%\item The \esp{} problem is FPT 
	%when 
	%parameterized by split vertex deletion set. 

    \item 
    %We design 
    A ($1+\epsilon$)-factor FPT approximation algorithm when parameterized by the feedback vertex set number.

	\item \esp{} is {\sf W}[2]-hard  parameterized by the chordal vertex deletion set.

\end{enumerate}
	\end{abstract}
  
	\section{Introduction}

Given a graph $G=(V,E)$ and a path $P$, the \emph{distance} from a vertex $v\in V(G)$ to $P$ is $\min\{d_G(v, w)\mid w\in V(P)\}$, where 
$d_G(v,w)$ is the distance between $v$ and $w$ in $G$. 
Given a graph $G$ and a path $P$,  the \emph{eccentricity} of $P$, denoted by ${\sf ecc}_G(P)$,  with respect to $G$  is defined as the maximum over all of the shortest distances between each vertex of  $G$ and $P$. 
Formally, ${\sf ecc}_G(P)=\max \{d_G(u, P) | u\in V(G)\}$. 
 Dragan and Leitert \cite{dragan2017minimum} introduced the problem of finding a shortest path with minimum eccentricity, called the \textsc{Minimum Eccentricity Shortest Path} problem (for short MESP) in a given undirected graph. They found interesting connections between MESP and the \textsc{Minimum Distortion Embedding} problem and 
obtained 
a better approximation algorithm for \textsc{Minimum Distortion Embedding}. 
MESP may be seen as a generalization of 
the \textsc{Dominating Path} Problem \cite{faudree2017minimum} that asks to find a path such that every vertex in the graph either belongs to the path or has a neighbor in the path. 
In MESP, the objective is to find a shortest path $P$ in $G$ such that the eccentricity of $P$ is minimum. 
Throughout the paper, 
we denote the minimum value over the eccentricities of all the shortest paths in $G$ as the \emph{eccentricity of the graph} $G$, denoted by ${\sf ecc}(G)$. 
MESP has applications in transportation planning, fluid transportation, water resource management, and communication networks. 
\begin{figure}[ht!]
	\centering
	\includegraphics[width=.8\textwidth]{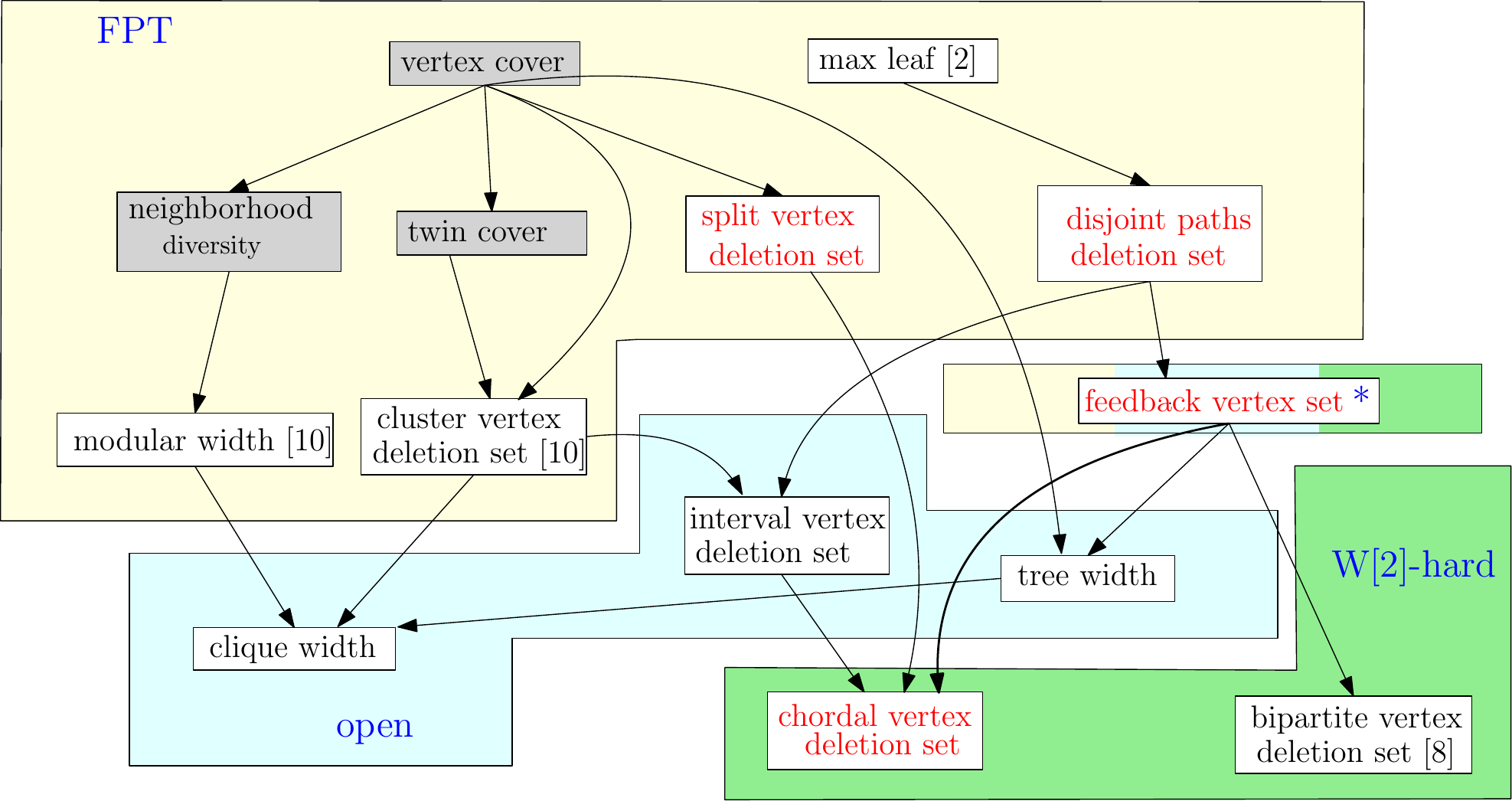}
	\caption{The hierarchy of parameters explored in this work. Arrow points from parameter $a$ to parameter $b$ indicates $b\leq f(a)$, for some computable function $f$. Parameters in red are studied in this paper. 
	The symbol ``$*$'' attached to the feedback vertex set means it is FPT in combination with the desired eccentricity. The grey box represents the result implied by those obtained. }
	\label{fig:parameter}
\end{figure}

\vspace{-3mm}

Dragan and Leitert \cite{dragan2015minimum}  demonstrated that fast algorithms for MESP imply fast approximation algorithms for \textsc{Minimum Line Distortion}, 
and the existence of low eccentricity shortest paths in special graph classes will imply low approximation bounds for those classes. %\todo[]{repetition of the distortion problem ?} 
They 
also showed that MESP is {\sf NP}-hard on planar bipartite graphs with maximum degree 3. 
In parameterized settings, they showed that MESP is {\sf W}[2]-hard for general graphs and gave an {\sf XP} algorithm for the problem when parameterized by eccentricity. 
Furthermore, they designed  2-approximation, 3-approximation, and 8-approximation algorithms for MESP running in time $O(n^3)$, $O(nm)$, and $O(m)$ respectively, where $n$ and $m$ represents the number of vertices and edges of the graph. 
The latter 8-approximation algorithm uses the double-BFS technique. 
In 2016, Birmel\'e et al.~\cite{birmele2016minimum} showed that the algorithm is, in fact, a 5-approximation algorithm by a deeper analysis of the double-BFS procedure and further extended the idea to get a 3-approximation algorithm, which still runs in linear time. 
Furthermore, they study the link between MESP and the laminarity of graphs introduced by Volk\'e et al.~\cite{volkel2016read} in which the covering path is required to be a diameter and established some tight bounds between MESP and the laminarity parameters. Dragan and Leitert \cite{dragan2015minimum} 
showed that MESP can be solved in linear time on distance-hereditary graphs
and in polynomial time on chordal and
dually chordal graphs. 
Recently, Ku\v cera and Such\'y \cite{kuvcera2021minimum} studied MESP with respect to some structural parameters and provided FPT algorithms for the problem
with respect to modular width, cluster vertex deletion (${\sf clvd}$), maximum leaf number, or the combined parameters 
disjoint
paths deletion (${\sf dpd}$) and eccentricity (${\sf ecc}$). 
We call the decision version of MESP, which is to check if there exists a shortest path $P$ 
such that for each $v\in V(G)$, the distance between $v$ and $P$ is at most $\ell$, as the \emph{Eccentricity Shortest Path} Problem (for short ESP). 
%which is to The decision version of the problem is
In this paper, we further extend the study of MESP in the parameterized setting. %\todo[]{Check if this definition needs to be introduced earlier ?}
\vspace{-3mm}

\subsection{Our Results and Discussion}
\vspace{-1mm}
In this paper, we study the parameterized complexity of ESP with respect to the 
%\todo[]{MESP or ESP ?} 
structural parameters: feedback vertex set (${\sf fvs}$), disjoint paths deletion set (${\sf dpd}$), split vertex deletion set (${\sf svd}$),  and chordal vertex deletion set (${\sf cvd}$). 
We call this version as \esp/$ \rho $, where $ \rho $ is the parameter. We now formally define 
\mespfvs 								(other problems can be defined similarly).

\defparprob{\mespfvs}{ An undirected graph $G$, a set $S\subseteq V(G)$ of size $k$ such that $G-S$ is a forest, and an integer $\ell$. }{$k+\ell$}{Does there exist a shortest path $P$ in $G$ such that for each $v\in V(G)$, $dist_G(v, P)\leq \ell$ ?}

First, we show an algorithm for \mespfvs, in Section \ref{sec:fvs}, that runs in $2^{\OO(k\log k)}\ell^kn^{\OO(1)}$ time where $\ell$ is the eccentricity of the graph and $k$ is the size of a feedback vertex set. In  Section \ref{sec-ptas}, 
   we design a ($1+\epsilon$)-factor FPT algorithm for \mespfappr.  Then, in Section \ref{sec:paths} we design an algorithm for 
   %Then, we show that There is an algorithm for 
   \mespdpd running in time $2^{\OO(k\log k)}\cdot n^{\OO(1)}$. 
   %   To eliminate the eccentricity parameter from the running time, we construct a set $Q$ of possible distance values of  disjoint paths deletion set $S$ to a solution path such that $|Q|$ is bounded by a function of $|S|$ which    helps us to obtain an FPT algorithm with    ${\sf dpd}$ only. 

In addition, we show that ESP/${\sf svd}$ admits an FPT algorithm. We then explore the problem on 
%chordal vertex deletion set 
${\sf cvd}$ which is a generalization of 
%each of 
${\sf fvs}$, ${\sf svd}$, ${\sf clvd}$ 
and show that ESP/${\sf cvd}$ is $W[2]$-hard. These results are 
%Due to space constraints, these results are 
presented in Sections  \ref{sec:split} and \ref{sec:hardness} respectively.

\section{Preliminaries}
%\noindent \textbf{Graph notations.}  \label{graphnotation}
 
 All the graphs considered in this paper are finite, unweighted, undirected, and connected. For standard graph notations, we refer to the graph theory book by R.~Diestel \cite{diestel2005graph}. For parameterized complexity terminology, we refer to the parameterized algorithms book by Cygan et al.~\cite{cygan2015parameterized}. 
% We use the following notations throughout the paper. Let $P$ be a path in $G$, then by $V(P)$, we denote the set of vertices in $P$. Consider a path $P$ in the graph $G$. Let $\cF$ be a set of paths in $G$, then by $V(\cF)$ we denote the set $\bigcup_{P\in \cF}V(P)$. We say that $P$ {\em covers} a vertex $v$ if there exists a vertex $u\in V(P)$ such that $d_{G}(v,u)\leq \ell$. The shortest distance of a vertex $v$ to a path $P$ is the minimum distance of a shortest $u\in V(P)$ to $v$ path, denoted by $d_G(v,P)$.
 For $n \in \mathbb{N}$,  we denote the sets $\{1,2,\cdots, n\}$ and $\{0,1,2,\cdots, n\}$ by $[n]$ and $[0,n]$ respectively.
For a graph $G=(V,E)$, we use $n$ and $m$ to denote the number of vertices and edges of $G$. 
%For $X \subseteq V(G)$, $G[X]$ denotes the subgraph of $G$ with vertex set $X$ and edge set $\{uv  \mid u,v \in X \mbox{ and } uv\in E(G)\}$, $G-X$ denotes the subgraph $G[V(G)\setminus X]$. For $Y \subseteq E(G)$, $G[Y]$ denotes the subgraph of $G$ with vertex set $\cup_{uv \in Y}\{u,v\}$ and edge set $Y$. The set of all the vertices adjacent to $ v $ is called the neighborhood of  $v$, and it is denoted by $N(v)$.
%Let $P$ be a path in $G$, then by $V(P)$, we denote the set of vertices in $P$. Consider a path $P$ in the graph $G$. For $v_1, v_\ell \in V(G)$, a $v_1$ to $v_\ell$-\emph{path} $P=(v_1,v_2,\cdots,v_{\ell-1},v_\ell)$ in $G$ is a sequence of (distinct) vertices, such that $V(P) \subseteq V(G)$ and for each $i \in [\ell-1]$, we have $v_iv_{i+1} \in E(G)$. Moreover, the edges in $\{v_iv_{i+1} \mid i \in [\ell-1]\}$ are called edges in $P$. The \emph{length} of a path is the number of edges in it. A {\em shortest $uv$-path} is a $u$ to $v$-path with the minimum number of edges. The \emph{distance} between two vertices $u$ and $v$ denoted by $d_G(u,v)$ is the length of the shortest $uv$-path.   For a vertex $ u \in G $ and a path $ P $ in $ G $, we use $d_G(u,P)$ to denote the length of the shortest $uv$-path among all the vertices $ v \in P $. Let $\cF$ be a set of paths in $G$, then by $V(\cF)$ we denote the set $\bigcup_{P\in \cF}V(P)$. 
Given an integer $\ell$, we say that a path $P$ {\em covers} a vertex $v$ if there exists a vertex $u\in V(P)$ such that the distance between the vertices $u$ and $v$, denoted by, $d_{G}(v,u)$, is at most $\ell$. 
 A \emph{feedback vertex set} of a graph $G$ is a set $S\subseteq V(G)$ such that $G-S$	is acyclic.

%\section{Ptreliminaries}
In addition to feedback vertex set, we have considered the following structural parameters:

\begin{enumerate}

	\item A \emph{disjoint paths deletion set} ($dpd$) is a set $ S\subseteq V(G)$ such that $ G-S$ is a set of disjoint paths. %A disjoint paths deletion set of size $k$ (if exists) can be found in $ \mathcal{O}(4^k n^2) $ time \cite{kuvcera2021minimum}.
	
	\item A \emph{split vertex deletion set} is a set $ S\subseteq V(G)$ such that $ G-S$ is a split graph, where a \emph{split graph} is a graph such that vertices of $G$ can be partitioned into two sets: one induces an independent set and other induce a clique. 
	
	\item A \emph{chordal vertex deletion set} is a set $ S\subseteq V(G)$ such that $ G-S$ is a chordal graph, where a chordal graph is a graph with no induced cycle of length at least 4. 
	
\end{enumerate}

Given a graph $G$, a feedback vertex et~\cite{cygan2015parameterized}, a disjoint path deletion set~\cite{kuvcera2021minimum}, a split vertex deletion set~\cite{cygan2013split}, and a chordal vertex deletion set~\cite{cao2016chordal},  of size $k$ can be found in FPT time parameterized by $k$.

%Next, we state a few results lemmas to which we refer multiple times in our paper.

Next, we state a few known results.

%\todo[]{Check the necessity of the following lemmas}
\begin{lemma} [Dragan and Leitert  \cite{dragan2017minimum}]\label{lem-fixed}
If a given graph $ G $ contains a shortest path with eccentricity $ \ell $, the MESP problem can be solved for $G$ in $\OO (n^{ 2\ell+2}m)$ time, where $ m $ is the number of edges in $ G $.
\end{lemma}
\begin{lemma}[ 
Ku{\v{c}}era and Such{\'y} \cite{kuvcera2021minimum} ]
\label{lem:permutation} 
For any graph $G=(V,E)$, any set $M\subseteq V$, and any vertex $s\in V$, at most one permutation $\pi=(m_1,\dots, m_{|M|})$ of the vertices in $M$ exists, such that, there is a shortest path $P$ with the following properties: The first vertex on $P$ is $s$,  $P$ contains all vertices of $M$, and the vertices from $M$ appear on $P$ in exactly the order given by $\pi$. Moreover, given a precomputed distance matrix for $G$, the permutation $\pi$ can be found in $\OO(|M|\log |M|)$ time. 
\end{lemma}

	%!TEX root = main.tex
\section{Parameterized by Feedback Vertex Set and Eccentricity} \label{sec:fvs}

In this section, we design an FPT algorithm for \mespfvs. 
The main theorem of this section is formally stated as follows. 
%~problem parameterized by the
%feedback vertex set number ({\sc $fvs$}) and the eccentricity of the input graph ({\sc $ecc$}). Specifically, we obtain the following theorem.

 \begin{sloppypar}
\begin{theorem}\label{thm:fvs}
	There is an algorithm for \mespfvs running in time $\OO(2^{\OO(k\log k)}\ell^kn^{\OO(1)})$. 
	%where $n$ is the number of vertices in the graph $G$. 
\end{theorem}

 \end{sloppypar}

%\subsection{Outline of the Algorithm}\label{sec:outline}

\noindent \textbf{Outline of the Algorithm.} Given a graph $G$ and a feedback vertex set $S$ of size $k$, 
we reduce \mespfvs to a ``path problem'' (which we call \amesp) on an auxiliary graph $G'$ (a forest) 
which is a subgraph of $G[V\setminus S]$, 
 using some reduction rules and two intermediate problems called \disjointmesp and \probenskel. In Section \ref{sec:outline2}, we show that \mespfvs and \disjointmesp are FPT-equivalent. Next, in Section \ref{sec:covering}, we reduce  \disjointmesp to  \probenskel. Then in Section \ref{sec:pathcover}, we reduce  \probenskel to  \amesp. Finally, in Section \ref{sec-path}, we design a dynamic programming based algorithm for \amesp that runs in 
%to show that the path problem is solvable in 
$\OO(\ell^2 2^{\OO(k\log k)}n^{\OO(1)})$ time. Together with the time taken for the reductions to the intermediate problems, 
we get our desired  FPT algorithm. A flow chart for the steps of the algorithm is shown in Figure \ref{fig:outline}.

\begin{figure}[ht!]
	\centering
	\includegraphics[width=1\textwidth]{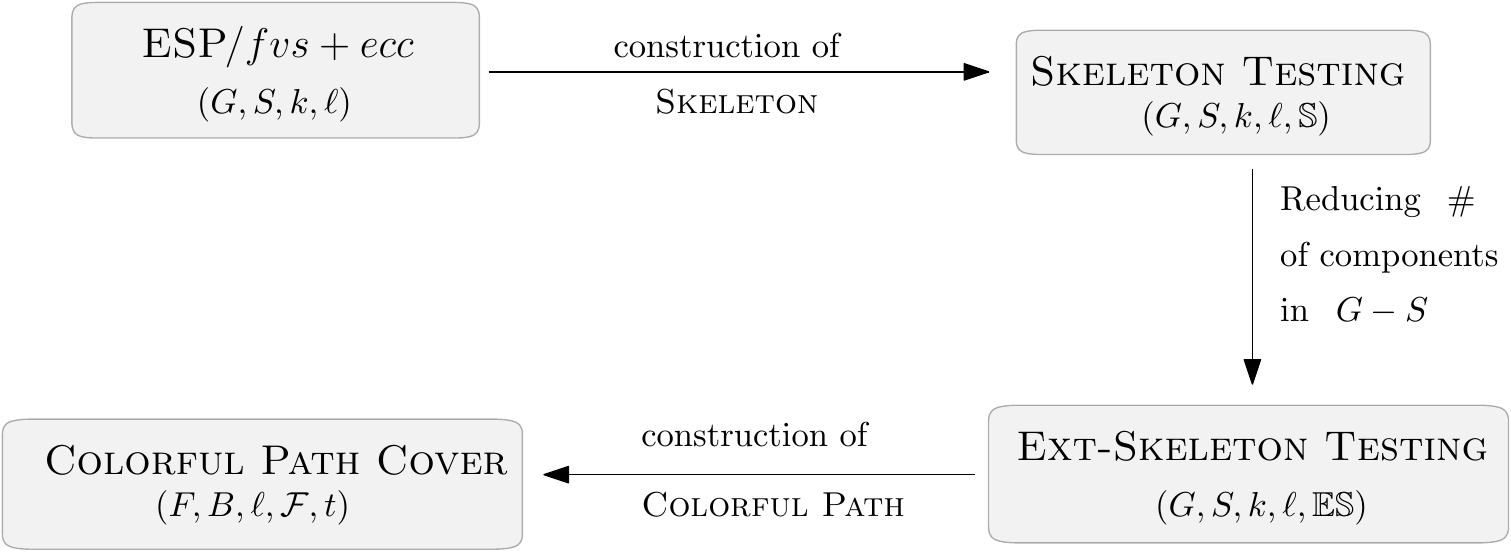}
	\caption{Flow chart of the Algorithm for \mespfvs.}
	\label{fig:outline}
\end{figure}

\vspace{-.75cm}

\subsection{Reducing to \disjointmesp}\label{sec:outline2}
%\todo[]{Need to Chnage the heading: Reducing to Skeleton Testing}
%Refer to Figure \ref{fig:outline} for an outline sketch of the algorithm.
The input to the problem is an instance $(G, S, k, \ell)$ where $S \subseteq V (G)$ is a
feedback vertex set of size $k$ in $G$. 
Let $(G,S, k,\ell)$ be a yes instance, and $P$ be a solution path 
which is a shortest path such that for each $v\in V(G)$, there exists $u\in V(P)$ such that $d_{G}(u,v)\leq \ell$.  
Our ultimate goal is to construct such a path $P$. 
%\todo[]{construct or check the existence ?}
Towards this, we try to get as much information as possible about $P$ in time $f(k,\ell) n^{\OO(1)}$. Observe that if $S$ is an empty set, then we can obtain $P$ by just knowing its end-points as there is a unique path in a tree between any two vertices. Generalizing this idea, given the set $S$, we define the notion of {\em skeleton} of $P$. 

 \begin{sloppypar}

\begin{definition}[Skeleton]\label{def:skeleton}
{\em A skeleton of $P$, denoted by $\skel$, is the following set of information. 
\begin{itemize}
\item End-vertices of $P$, say $u, v\in V(G)$. 
\item A subset of $S\setminus \{u,v\}$, say $M$, of vertices that appear on $P$. That is, 
$V(P)\cap (S\setminus \{u, v\})=M$. 
%Let $|M|=\alpha$
\item The order in which the vertices of $M$ appear on $P$, is given by an ordering $\pi=m_1,m_2, \dots, m_{|M|}$. For notational convenience, we denote $u$ by $m_0$ and $v$ by $m_{|M|+1}$. 
\item A distance profile $(f,g)$ for the set $X= S\setminus M$, 
%of the set of vertices of $S$ that do not appear on $P$, that consist of two functions
 is defined as follows: The function 
 $f:X\rightarrow [\ell]$ such that $f(x)$ denotes the shortest distance of the vertex $x$ from $P$, and the function $g:X\rightarrow \{0,1,\cdots, |M|+1,(0,1),(1,2),\cdots, (|M|,|M|+1)\}$ such that $g(x)$ stores the information about the location 
 of the vertex on $P$, that is closest to $x$. That is, if the vertex closest to $P$ belongs to $\{m_0,m_1, \dots, m_{|M|},  m_{|M|+1}\}$ 
 then $g(x)$ stores this by assigning the corresponding index. 
 %\todo{If there are multiple of them ?} 
 Else, the closest vertex belongs to the path segment between $m_{i},m_{i+1}$, for some $0\leq i\leq |M|$, which $g(x)$ stores by assigning $(i,i+1)$.

\end{itemize}
}
\end{definition}

 \end{sloppypar}

\begin{figure}[t!]
	\centering
	\includegraphics[width=.8\textwidth]{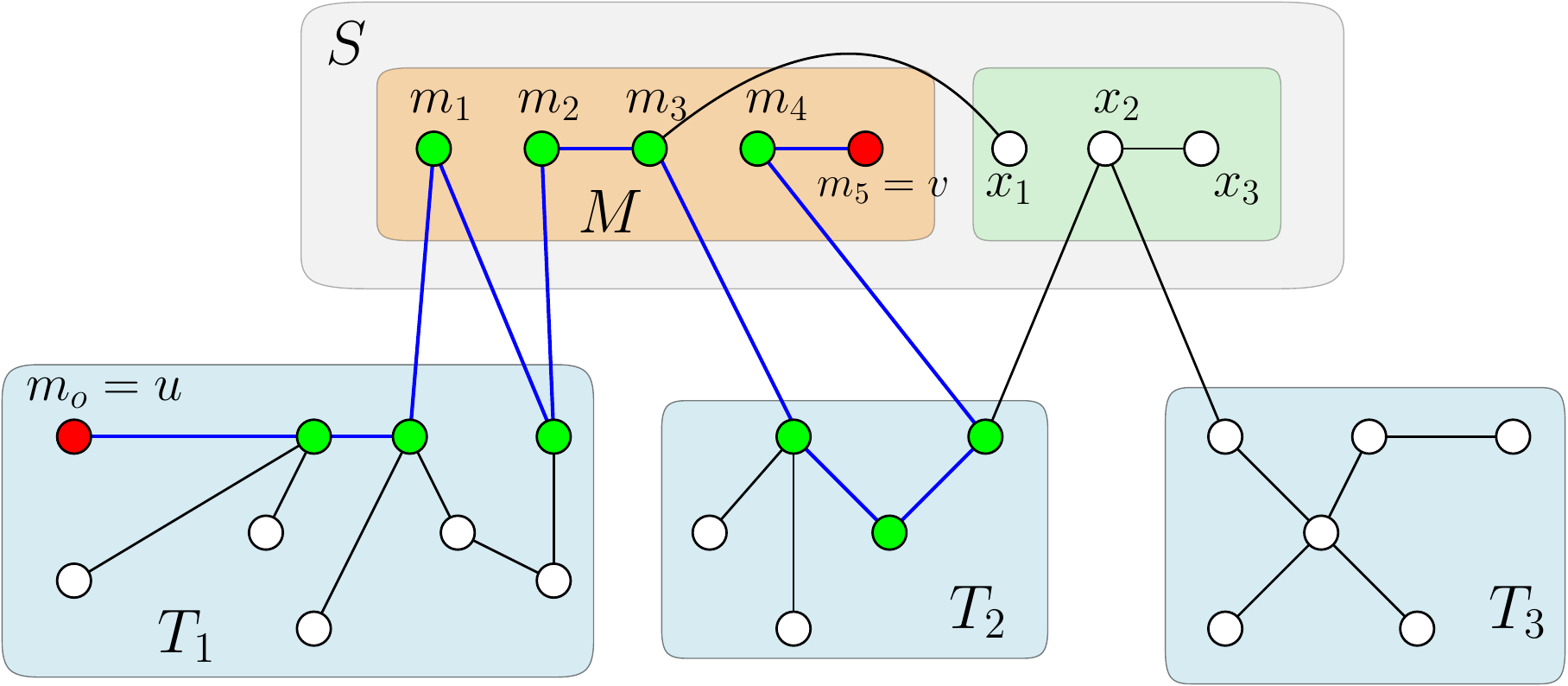}
	\caption{Example of a skeleton of $P$. Here $ P $ is a shortest path (blue edges) between two red colored vertices $ u $ and $ v $  through green colored internal vertices. For the vertices $x_1, x_2$ and $x_3$, 
 %we have 
 $f(x_1)= f(x_2)=1$,   
 $f(x_3)=2$,  
  $g(x_1)=3$,  $g(x_2)= g(x_3)=(3,4). $ }
	\label{fig:skeleton}
 \vspace{-.5cm}
\end{figure}

An illustration of a skeleton is given in Fig. \ref{fig:skeleton}. 
By following the definition of skeletons, we get an upper bound on them. 

\begin{observation}
\label{lem:skelnum}
%If $(G,S, k,\ell)$ is a yes-instance then 
The number of skeletons is upper bounded by $n^2 2^k  k!  \ell^k  (2k+2)^k$. 
\end{observation}

\noindent 
We say that a path $P$ {\em realizes}  a skeleton $\skel$ if the following holds. 
\vspace{-0.2cm}
\begin{enumerate}
     \item $M=S\cap  V(P)$,  $X\cap V(P)=\emptyset$, 
      the ordering of vertices in $M$ in $P$ is equal to $\pi$, 
      endpoints of $P$ are $m_0$ and $m_{|M|+1}$, 
      
      \item For each $v\in V(G)$, there exists a vertex $u\in V(P)$ such that $d_{G}(u,v)\leq \ell$, 
      
      \item For each $v\in X$, 
      $d_G(v,w)\geq f(v)$  for all $w\in V(P)$ (where $f(v)$ is the shortest distance from $v$ to any vertex on $P$ in $G$), and
      
      \item For each $v\in X$, if $g(v)=i$, where $i\in [0,|M|+1]$, then $d_G(v,m_i)=f(v)$ 
      and if $g(v)=(i,i+1)$ where 
      $i\in [0,|M|]$, then there exists a vertex $u$ on a subpath $m_i$ to $m_{i+1}$ in $P$ such that  $u\notin \{m_i,m_{i+1}\}$ and $d_G(u,v)=f(v)$. 
\end{enumerate}

\noindent Now, given an input $(G,S, k,\ell)$ and a skeleton $\skel$, our goal is 
to test whether the skeleton can be {\em realized} into a desired path $P$. This leads to the following problem.

\defparprob{\disjointmesp }
{A graph $G$, a set $S\subseteq V(G)$ of size $k$ such that $G-S$ is a forest, an integer $\ell$, and a skeleton $\skel$. } 
{$k+\ell$}
{Does there exist a shortest path $P$ in $G$ that realizes $\skel$? 
}

%Our next lemma formally proves that \mespfvs and \disjointmesp are FPT-equivalent. That is,  \mespfvs is FPT if and only if  \disjointmesp is FPT. 

Our next lemma shows a reduction from \mespfvs to \disjointmesp problem. 

\begin{lemma}\label{lem:mespfvscorr}
	$(G,S,k,\ell)$ is a yes instance of \mespfvs, if and only if there exists a skeleton $\skel$ such that 
	$(G,S,k,\ell,\skel)$ is a yes instance of \dmesp. 
\end{lemma}
%%%%%%%%%%%%%%%%%%%%%Added-12Jan-2023
%\section{Proof of Lemma \ref{lem:mespfvscorr} from Section \ref{sec:outline}}\label{app:sec3.1}
\begin{proof}
	For the forward direction, suppose that $(G,S,k,\ell)$ is a yes instance of \mespfvs~and let $P$ be its solution path. Let $M= V(P)\cap S$, $|M|=t$ and $X=S\setminus M$. Let $\pi$ be the ordering of vertices in $M$ on $P$. Let $m_0,m_{|M|+1}$ be the first and last endpoints of $P$, respectively. We define $f:X\rightarrow [\ell]$ such that for a vertex $v\in X$,  $f(v)$ is the minimum distance from $v$ to path $P$. Let $u$ be the vertex in path $P$ such that $d(v,u)=f(v)$. Note that $f$ only assigns values from the set $[\ell]$. We define a function $g$, $g:X\rightarrow \{0,1,\cdots, t+1,(0,1),(1,2),\cdots, (t,t+1)\}$ as follows: $g(v)=i$ if $u=m_i$ for some $i\in [0,t+1]$,
	and $g(v)=(i,i+1)$ if $u$ is contained in $m_i$ to $m_{i+1}$ subpath for some $i\in [0,t]$. 
	Note that $g$ only assigns values from the set $\{0,1,\cdots, t+1,(0,1),(1,2),\cdots, (t,t+1)\}$. Observe that the tuple $(M,X,\pi,m_0,m_{t+1},f,g)$ is a skeleton of $P$. Hence, $(G,S,k,\ell,\skel)$  is a yes instance of \dmesp, where $\skel=(M,X,\pi,m_0,m_{t+1},f,g)$.
	
	In the backward direction, suppose that $\skel=(M,X,\pi,m_0,m_{t+1},f,g)$ is a skeleton such that $(G,S,k,\ell,\skel)$ is a yes instance of \dmesp and let $P$ be its solution. Observe that $P$ is also a solution to  $(G,S,k,\ell)$ of \mespfvs, as for every  $v\in V(G)$, there is a vertex $u\in V(P)$ such that $d_G(u,v)\leq \ell$.
\qed

\end{proof}

\noindent Observation~\ref{lem:skelnum} upper bounds the number of skeletons by $2^{\OO(k(\log k + \log \ell))}n^2$. This together with Lemma~\ref{lem:mespfvscorr}, implies that \mespfvs and \disjointmesp are FPT-equivalent. Thus, from now onwards, we focus on \disjointmesp. %\todo[]{Remove this from here. Let us keep the proof of Theorem \ref{thm:fvs} in the main text and it describes this paragraph. }

	%!TEX root = main.tex
\subsection{Algorithm for \disjointmesp}\label{sec:disjoint}%\todo[]{Change the heading: Reducing to \amesp}

\sloppy
%Next, we design our algorithm for \dmesp. 
%Consider an instance $\cD=(G,S,k,\ell,M,X,\pi,m_0,m_{|M|+1},f,g)$ of \dmesp. 
Let $(G,S,k,\ell,\skel)$ be an instance of  \dmesp, where $\skel=(M,X,\pi,m_0,m_{|M|+1},f,g)$.  Our algorithm works as follows. First, the algorithm performs a simple sanity check by reduction rule. In essence, it checks whether the different components of the skeleton $\skel$ are valid. 
% This allows us to bound the number of connected components in $G-S$. 

\begin{reduction rule}[Sanity Test $1$] \label{rr:dmesp1}
 Return that $(G,S,k,\ell,\skel)$ is a no instance of \dmesp, if one of the following holds:
\begin{enumerate}
    \item 
    %{\bf (1)} 
    For  $i\in [0,|M|]$, $m_im_{i+1}$ is an edge in $G$ and $g^{-1}((i,i+1))\neq \emptyset$. ({\sf $g$ is not valid.})
    \item 
    %{\bf (2)}  
    For a vertex $v\in X$, there exists a vertex $u\in M\cup \{m_0,m_{|M|+1}\}$ such that  $d_G(u,v)<f(v)$.   ({\sf $f$ is not valid.})
    \item 
    %{\bf (3)} 
    For a vertex $v\in X$, $g(v)=i$ and $d_G(v,m_i)>f(v)$.  ({\sf $f$ is not valid.})

    \item 
    %{\bf (4)}  
    For an $i\in [0,|M|]$, $m_im_{i+1}$ is not an edge in $G$, and there is either no $m_i$ to $m_{i+1}$ path in $G- (S\setminus\{ m_i, m_{i+1}\})$ or the length of the path is larger than the shortest path length of  $m_i$ to $m_{i+1}$ path in $G$. ({\sf $\pi$ is not valid.})
    \item 
    %{\bf (5)}  
    For $i,j\in [0,|M|]$, $i<j$, there exists $m_i$ to $m_{i+1}$ shortest path $P_i$ in $G-(S\setminus\{m_i,m_{i+1}\})$ and a $m_j$ to $m_{j+1}$ shortest path $P_j$ in $G-(S\setminus\{m_j,m_{j+1}\})$ such that if $j=i+1$, then  $(V(P_i)\setminus \{m_{i+1}\})\cap (V(P_j)\setminus \{m_{j}\})\neq \emptyset$, otherwise $V(P_i)\cap V(P_j)\neq \emptyset$.  ({\sf $\pi$ is not valid -- shortest path claim will be violated.})

    \item 
    %{\bf (6)} 
     For  $i\in [0,|M|]$ such that $m_im_{i+1}\notin E(G)$, $g^{-1}((i,i+1))\neq \emptyset$, and for every connected component $C$ in $G-S$, and for every $m_i$ to $m_{i+1}$ path $P$ in $G[V(C)\cup \{m_i,m_{i+1}\}]$ there exists a vertex $u\in g^{-1}((i,i+1))$ such that 
     there is no vertex $v\in V(P)\setminus \{m_i,m_{i+1}\}$ for which $d_G(u,v)=f(u)$. ({\sf $g$ is not valid.})
     
\end{enumerate}

%  \end{enumerate}
 \end{reduction rule}
 
  \begin{lemma}\label{lem:rr1safe}
	Reduction rule \ref{rr:dmesp1}~is safe.
\end{lemma}

\begin{proof}
	Proof of (1) follows from the fact that, to maintain the shortest path property of the solution, the edge $m_im_{i+1}$ must be in solution. It contradicts that for every $v\in g^{-1}((i,i+1))$, there exists a vertex $u$ in $m_i$ to $m_{i+1}$ subpath of the solution path such that $u\notin \{m_i,m_{i+1}\}$ and $d_G(u,v)=f(v)$. Proof of (2) follows from the fact that otherwise, the property of the function $g$ will be violated in the solution path. Proof of (4) follows from the shortest path property of the solution. Proofs of (3) and (6) are trivial as it is impossible to find a solution path in these cases. 
	
	Proof of (5): Suppose that  $(G,S,k,\ell,\skel)$ is a yes instance of \dmesp and $P$ be its solution. For $i,j\in [0,|M|]$, $i<j$, there exists $m_i$ to $m_{i+1}$ shortest path $P_i$ in $G-(S\setminus\{m_i,m_{i+1}\})$ and a $m_j$ to $m_{j+1}$ shortest path $P_j$ in $G-(S\setminus\{m_j,m_{j+1}\})$. Suppose that $j=i+1$, then  $(V(P_i)\setminus \{m_{i+1}\})\cap (V(P_j)\setminus \{m_{j}\})\neq \emptyset$. Let $v^*\in (V(P_i)\setminus \{m_{i+1}\})\cap (V(P_j)\setminus \{m_{j}\})$ be the last intersecting vertex on paths $P_i, P_j$  by ordering on path  $P_i$. Observe that distance of $m_i$ to $v^*$ in $P_i$ is strictly less that $m_i$ to $m_{i+1}(m_j)$ subpath on $P$, that is $d_{P_i}(m_i,v^*)<d_P(m_i,m_{i+1})$, as $v^*\neq m_{i+1}=m_j$ and both $P_i$ and $m_i$ to $m_{i+1}$ subpath on $P$ are shortest $m_i$ to $m_{i+1}$ paths in $G$. Similarly, $d_{P_j}(v^*,m_{j+1})\leq d_P(m_j,m_{j+1})$, as  both $P_j$ and $m_j$ to $m_{j+1}$ subpath on $P$ are shortest $m_j$ to $m_{j+1}$ paths in $G$. This implies that $d_P(m_i,m_{j+1})>d_{P_i}(m_i,v^*)+d_{P_j}(v^*,m_{j+1})$. By replacing $m_i$ to $m_{j+1}$ subpath in $P$ by the subpaths $m_i$ to $v^*$ of $P_i$ and $v^*$ to $m_{j+1}$ of $P_j$ we obtain a shorter path in $G$, which contradicts the shortest path property of the solution $P$. The case when $j\neq i+1$ can be argued analogously.	\qed
 \end{proof}

\noindent
\textbf{Reducing the components of $G-S$:} 
Now, we describe our marking procedure and reduction rules that are applied 
on the connected components in $G-S$. Let $P_i$ be a path segment (subpath) of $P$,  between  $m_i$ and $m_{i+1}$, with at least two edges. Further, let $P_i^{int}$ be the subpath of $P_i$, obtained by deleting $m_i$ and $m_{i+1}$. Then, we have that $P_i^{int}$ is a path between two vertices in $G-S$ (that is, a path in the forest $G-S$). This implies that $P$ is made up of $S$ and at most $k+1$ paths of forest $G-S$. Let these paths be ${\mathbb P}=P_1^{int},\ldots,P_q^{int}$, where $q\leq k+1$. Next, we try to understand these $k+1$ paths of forest $G-S$. Indeed, if there exists a component $C$ in $G-S$ such that it has a vertex that is far away from every vertex in $S$, then $C$ must contain one of the paths in $\mathbb{P}$ ({\em essential components}). The number of such components can be at most $k+1$. The other reason that a component contains a path from $\mathbb{P}$ is  to select a path that helps us to satisfy constraints given by the  $g$ function 
({\em $g$-satisfying components}).  Next, we give a procedure that marks $\OO(k)$ components, and later, we show that all unmarked components can be safely deleted. 

\medskip

%Using these, we bound the number of components as a function linear in $k$. 

\noindent{\bf Marking Procedure:} Let ${\cal C}^*$ be the set of marked connected components of $G-S$. Initially, let $\cC^*=\emptyset$. 
\begin{itemize}
    \item \noindent{\bf Step 1.} If there exists a connected component $C$ in $G-(S\cup V(\cC^*))$, such that it contains a vertex $v$ with $d_G(v,m_i)>\ell$, for all $m_i\in M$, 
    and $d_G(v,u)>\ell-f(u)$, for all $u\in X$, then add  $C$ to $\cC^*$. ({\sf Marking essential components})
  \item \noindent{\bf Step 2.} For $i=0$ to $|M|$ proceed as follows: Let $C$ be some connected component in $G-(S\cup V(\cC^*))$ such that there exists a $m_i$ to $m_{i+1}$ path $P_i$ in $G[V(C)\cup \{m_i,m_{i+1}\}]$, which is a shortest  $m_i$ to $m_{i+1}$ path in $G$ and for every vertex $v\in g^{-1}((i,i+1))$, there exists a vertex $u\in V(P_i)\setminus \{m_i,m_{i+1}\}$ for which $d_G(u,v)=f(v)$. Then, add  $C$ to $\cC^*$ and increase the index $i$.  ({\sf Marking $g$-satisfying components})
\end{itemize}

Let $\cC_1$ be the set of connected components added to $\cC^*$ in Step 1. We now state a few reduction rules the algorithm applies exhaustively in the order in which they are stated.

%\qed

%\begin{sloppypar}

\vspace{-0.2cm}
 \begin{sloppypar}
 \begin{reduction rule}\label{rr:dmesp2}
 If $|\cC_1|\geq k+2$, then return that $(G,S,k,\ell,\skel)$ 
 is a no instance of \dmesp.
 \end{reduction rule}
  \end{sloppypar}

\vspace{-0.5cm}
\begin{lemma}\label{lem:mark1}
	Reduction rule \ref{rr:dmesp2}~is safe. 
	%If $|\cC_1|\geq k+2$, then $(G,S,k,\ell,M,X,\pi,m_0,m_{|M|+1},f,g)$ is a no instance of \dmesp
\end{lemma}
\begin{proof}
	For each component $C$ in $\cC_1$, $C$ contains a vertex $v$ such that $d_G(v,m_i)>\ell$, for all $m_i\in M$ and $d_G(v,u)>\ell-f(u)$, for all $u\in X$, which implies we must add a path from component $C$ in solution path as a subpath such that it contains a vertex that covers $v$. Observe that we can add at most $|M|+1$ subpaths in the solution path. Therefore, $|\cC_1|\leq |M|+1$ if 
	%$(G,S,k,\ell,M,X,\pi,m_0,m_{|M|+1},f,g)$ 
	$(G,S,k,\ell,\skel)$  
	is a yes instance of \dmesp. We obtain the required bound as $|M|\leq k$. \qed \end{proof}

  \begin{sloppypar}
 \begin{reduction rule}\label{rr:dmesp3}
 If there exists a connected component $C$ in $G-S$ such that $C\notin \cC^*$, then delete $V(C)$ from $G$. The resultant instance is $(G-V(C),S,k,\ell,\skel)$.  
 % $(G-V(C),S,k,\ell,M,X,\pi,m_0,m_{|M|+1},f,g)$. 
 \end{reduction rule}

\end{sloppypar}

\begin{lemma}\label{lem:rr3}
	Reduction rule \ref{rr:dmesp3}~is safe.
\end{lemma}
\begin{proof}
    We show that $(G,S,k,\ell,\skel)$ is a yes instance of \dmesp if and only if $(G-V(C),S,k,\ell,\skel)$ is a yes instance of \dmesp. Recall that $\skel=(M,X,\pi,m_0,m_{|M|+1},f,g)$. In the forward direction, consider that 
	%$(G,S,k,\ell,M,X,\pi,m_0,m_{|M|+1},f,g)$ 
	$(G,S,k,\ell,\skel)$ is a yes instance of \dmesp and $P$ be its solution. Suppose that $P$ contains a $m_i$ to $m_{i+1}$ subpath $P^*$ such that $P^*$ is a path in $G[V(C)\cup \{m_i,m_{i+1}\}]$. Since $C$ is not in $\cC^*$, there exists a connected component $C'\in \cC^*$ such that there exists a $m_i$ to $m_{i+1}$ path $P'$ in $G[V(C')\cup \{m_i,m_{i+1}\}]$, such that $P'$ is a shortest  $m_i$ to $m_{i+1}$ path in $G$ and for every vertex $v\in g^{-1}((i,i+1))$, $d_G(v,P'-\{m_i,m_{i+1}\})=f(v)$. We replace subpath  $P^*$ from $P$ by subpath $P'$. Let $P''$ is the resultant path. We claim that $P''$ is also a solution to 
	%$(G,S,k,\ell,M,X,\pi,m_0,m_{|M|+1},f,g)$
	$(G,S,k,\ell,\skel)$. 
	Observe that $P''$ satisfies functions $f$ and $g$. Suppose that there is a vertex $v\in V(G)$ such that $d_G(v,P'')> \ell$. Then $v$ must be in $C$. Since $C$ is not added to $\cC^*$ in Step 1 of the marking procedure, every vertex in $C$ is either at distance at most $\ell$ from some $m_i,i\in [0,|M|+1]$ or $d_G(v,u)\leq \ell-f(u)$ for some $u\in X$. Hence $v$ is covered by some vertex on $P''$, a contradiction.
	So $ P'' $ is also a solution to $(G, S, k, \ell, \skel)$ and $P''$ does not contain any vertex from $C$. As every vertex of $C$ gets covered,   $P''$  is also a solution to $(G - V (C), S, k, \ell, \skel)$.	If $ P $ contains no vertex from $ C $ then by the similar procedure   we can obtain  a solution path $\hat{P}$ for %$(G,S,k,\ell,M,X,\pi,m_0,m_{|M|+1},f,g)$ 
	$(G,S,k,\ell,\skel)$ which doesn't contain any vertex of $C$. Observe that $\hat{P}$ is also a solution to $(G-V(C),S,k,\ell,\skel)$ 
	%$(G-V(C),S,k,\ell,M,X,\pi,m_0,m_{|M|+1},f,g)$ 
	of \dmesp. 
	
	In the backward direction suppose that 
	$(G-V(C),S,k,\ell,\skel)$ 
	%$(G-V(C),S,k,\ell,M,X,\pi,m_0,m_{|M|+1},f,g)$
	is a yes instance of \dmesp and let $P$ be its solution. Since $C$ is not added to $\cC^*$ in Step 1 of the marking procedure, every vertex in $C$ is either  distance at most $\ell$ from some $m_i,i\in [0,|M|+1]$ or $d_G(v,u)\leq \ell-f(u)$ for some $u\in X$, therefore every vertex in $C$ is covered by some vertex on $P$. This implies that $P$ is also a solution to 
	%$(G,S,k,\ell,M,X,\pi,m_0,m_{|M|+1},f,g)$
	$(G,S,k,\ell,\skel)$. This completes the proof. 
\qed
\end{proof}
%%%%%%%%%%%%%%%%%%%%%%%%%%

%\noindent Due to space constraint, we present the proof in Appendix \ref{app:sec3.2.1}. 
%Due to space constraints, we present the proof of Lemma \ref{lem:rr3} in Appendix \ref{app:sec3.2.1}.
Observe that when Reduction rule \ref{rr:dmesp2}~and Reduction rule \ref{rr:dmesp3}~are no longer applicable, the number of connected components in $G-S$ is bounded by $2(k+1)$. This is because $|\cC_1|\leq k+1$ and 
there exists a path (that is part of the solution) from each component in $\cC^*-\cC_1$ and therefore $|\cC^*-\cC_1|\leq k+1$. 
%\todo{rewrite this a bit} 
Otherwise, the given instance is a no instance of \dmesp. 
Notice that all our reduction rules can be applied in $n^{\OO(1)}$ time.

 \begin{sloppypar}

  \vspace{-3mm}

%\noindent \textbf{Reducing to a covering problem on forest:} \label{sec:covering}%\todo[]{Change the heading to Reducing Skeleton to Ext-Skeleton?}

\subsection{Reducing \dmesp to \probenskel:} \label{sec:covering}%\todo[]{Change the heading to Reducing Skeleton to Ext-Skeleton?}
 \vspace{-1mm}
 Let $(G,S,k,\ell,\skel)$ be a reduced instance of  \dmesp. That is, an instance on which Reduction Rules~ \ref{rr:dmesp1}, \ref{rr:dmesp2}~and \ref{rr:dmesp3}~are no longer applicable. This implies that the number of connected components in $G-S$ is at most $2k+2$. Next, we enrich our skeleton by adding a function $\gamma$, which records  an index of a component in $G-S$ that gives the $m_i$ to $m_{i+1}$ subpath in $P$ or records that $m_im_{i+1}$ is an edge in the desired path $P$, where $i\in [0,|M|]$.
 \vspace{-1mm}

\begin{definition}[Enriched Skeleton] \label{def:enriched}
 An \emph{enriched skeleton} of a path $P$, denoted by $\enskel$, contains $\skel$ and a
 segment profile of paths between $m_i$ and $m_{i+1}$, for $i\in[0,M]$. Let  $C_1,C_2,\dots, C_{q}$ be the connected components in $G-S$. Then, the segment profile is given by a function $\gamma:[0,|M|]\rightarrow [0,q]$. 
The function $\gamma$ represents the following: For each $i\in [0,|M|]$, if $\gamma(i)=0$, then the pair $m_i,m_{i+1}$ should be connected by an edge in the solution path $P$, 
otherwise if $\gamma(i)=j$, then in $P$, 
the $m_i$ to $m_{i+1}$ subpath is contained in $G[V(C_j)\cup \{m_i,m_{i+1}\}]$. 
% \end{itemize}
% $\enskel$ is said to be enriching $\skel$. 
% }
Also, $\enskel$ is said to be enriching the skeleton $\skel$.
\end{definition}
\vspace{-1mm}

Let $\skel$ be a skeleton. The number of $\enskel$, that enrich  $\skel$ is upper bounded by $(q+1)^{k+1}$. Thus, this is not useful for us unless $q$ is bounded by a function of $k,\ell$. Fortunately, the number of connected components in $G-S$ is at most $2k+2$, and thus the number of $\enskel$ is upper bounded by $2^{\OO(k\log k)}$. 

We say that a path $P$ {\em realizes}  an enriched skeleton $\enskel$ enriching $\skel$, if $P$ realizes $\skel$ and satisfies $\gamma$. Similar to \dmesp, we can define \probenskel, where the aim is to test if a path exists that realizes an enriched skeleton $\enskel$. Further, it is easy to see that \dmesp and \probenskel are FPT-equivalent, and thus we can focus on \probenskel. Let $(G,S,k,\ell,\enskel)$ be an instance of  \probenskel, where $G-S$ has at most $2k+2$ components. Similarly, as  \dmesp, we first apply some sanity testing on an instance of  \probenskel. 

\begin{reduction rule}[Sanity Test $2$] \label{rr:sanitytest2}
 Return that $(G,S,k,\ell,\enskel)$ is a no instance of \probenskel, if one of the following holds:
\begin{enumerate}
    \item $m_im_{i+1}$ is an edge in $G$ and $\gamma(i)\neq 0$, (or) $m_im_{i+1}$ is not an edge in $G$ and $\gamma(i)= 0$. 
    \item For an $i\in [|M|]$, $\gamma(i)=j\neq 0$ and there is, 
    \begin{itemize}
        \item No $m_i$ to $m_{i+1}$ path in $G[V(C_j)\cup \{m_i,m_{i+1}\}]$, (or) 
        \item No $m_i$ to $m_{i+1}$ path in $G[V(C_j)\cup \{m_i,m_{i+1}\}]$ which is also a shortest $m_i$ to $m_{i+1}$ path in $G$, (or) 
        \item There does not exist a $m_i$ to $m_{i+1}$ path $P_{i}$ in $G[V(C_j)\cup \{m_i,m_{i+1}\}]$ which is also a shortest $m_i$ to $m_{i+1}$ path in $G$ and satisfies the property that for every vertex $v\in g^{-1}((i,i+1))$, there exists a vertex $u\in V(P_{i})\setminus \{m_i,m_{i+1}\}$ for which $d_G(u,v)=f(v)$.

    \end{itemize}
\end{enumerate}
\end{reduction rule}

The safeness of the above rule follows from Definition \ref{def:enriched}.

%\vspace{0.3cm}
%\noindent{\bf Reduction to the instances of \amesp.}  

\vspace{-3mm}
%\medskip

\subsection{Reducing  \probenskel to \amesp}\label{sec:pathcover}
%\todo[]{Add a subsection: Reduce from Ext-skeleton to Colorful path cover} 
Let $(G,S,k,\ell,\enskel)$ be an instance of  \probenskel on which Reduction Rule~\ref{rr:sanitytest2} is no longer applicable. Further, let us assume that the number of components in $G-S$ is $k'\leq 2k+2$ and $\gamma:[0,|M|]\rightarrow [0,k']$ be the function in $\enskel$. Our objective is to find a path $P$ that realizes $\enskel$. Observe that for an $i\in [0,|M|]$, if $\gamma(i)=j\neq 0$, then the {\em interesting} paths to connect $m_i,m_{i+1}$ pair are contained in component $C_j$ in $G-S$. Moreover, among all the paths that connect $m_i$ to $m_{i+1}$ in $C_j$, only the shortest paths that satisfy the function $g$ are the interesting paths. Therefore, we enumerate all the feasible paths for each $m_i,m_{i+1}$ pair in a family $\cF_i$ and focus on finding a solution that contains subpaths from this enumerated set of paths only. Notice that now our problem is reduced to finding a 
 set of paths $\cP$ in $G-S$ which contains exactly one path from each family of feasible paths and covers all the vertices in $G-S$ which are far away from $S$. In what follows, we formalize the above discussion. First, we describe our enumeration procedure.

 For each $i\in [0,|M|]$ where $\gamma(i)=j\neq 0$, we construct a family $\cF_i$ of {\em feasible paths} as follows. Let $P_i$ be a path in $G[V(C_j)\cup \{m_i,m_{i+1}\}]$, such that
   (i) $P_i$ is a shortest  $m_i$ to $m_{i+1}$ path in $G$, (ii) for every vertex $v\in g^{-1}((i,i+1))$, $d_G(v,P_i-\{m_i,m_{i+1}\})=f(v)$. Let $m'_i, m'_{i+1}$ be the neighbours of $m_i,m_{i+1}$, respectively in $P_i$. Then we add $m'_i$ to $m'_{i+1}$ subpath to $\cF_i$.  Observe that a family $\cF_i$ of feasible paths satisfies the following properties:
   (1)  $V(\cF_i)\cap V(\cF_{i'})=\emptyset$, for all $i,{i'}\in \gamma^{-1}(j), i\neq i'$, as item 5 of  \cref{rr:dmesp1}~is not applicable, and we add only shortest paths in families. (2) $\cF_i$ contains paths from exactly one component in $G-S$ (by the construction).  Let $\cF$ be the collection of all the families of feasible paths. 
   
   The above discussion leads us to the following problem.

\defproblem{\amesp}
{A forest $F$, 
a set $B\subseteq V(F)$,  an integer $\ell$, and a family $\cF=\{\cF_1,\cF_2,\dots ,\cF_t\}$ of $t$ disjoint families of feasible paths.}
{Is there a set $\cP$ of $t$ paths such that for each $\cF_i$, $i\in [t]$, $|\cP \cap \cF_i|=1$ and for every vertex $v\in B$, there exists a path $P\in \cP$ and a vertex $u\in V(P)$, such that $d_F(u,v)\leq \ell$?}

 Let $F$ be the forest obtained from $G-S$ by removing all the components $C_j$ in $G-S$ such that $\gamma^{-1}(j)=\emptyset$, that is, components which do not contain any interesting paths. Notice that the number of components that contain interesting paths is at most $2k+2$. 
 We let $B\subseteq V(F)$ be the set of vertices which is not covered by vertices in $S$, that is, it contains all the vertices $v\in V(F)$ such that $d_G(v,m_i)>\ell$, for all $i\in [0,|M|+1]$ and $d_G(v,u)>\ell-f(u)$, for all $u\in X$. We claim that it is sufficient to solve \amesp on instance $(F,B,\ell,\cF)$ where $F$ consists of at most $2k+2$ trees. The following lemma shows a reduction formally and concludes that \probenskel parameterized by $k$ 
 %\todo[]{Should be just $k$?}
 and \amesp problem parameterized by $k$, are FPT-equivalent. 
%Due to space constraints, we present the proof of Lemma \ref{lem:disjointfvscorr} in Appendix \ref{app:sec3.2.2}. 
\vspace{-.5mm}

\begin{lemma}\label{lem:disjointfvscorr}
$(G,S,k,\ell,\enskel)$ is a yes instance of \probenskel if and only if $(F,B,\ell, \cF)$ is a yes instance of \amesp. 
\end{lemma}
\begin{proof}
	Recall that $\skel=(M,X,\pi,m_0,m_{|M|+1},f,g)$ and $\enskel=(\skel,\gamma)$. In the forward direction, suppose that $(G,S,k,\ell,\enskel)$ is a yes instance of \dmesp and let $P$ be its solution. Consider a connected component $C_j$ in $G-S$. If $\gamma^{-1}(j)=\emptyset$, then by the properties of $P$, there does not exists $i\in [0,|M|]$ such that $m_i$ to $m_{i+1}$ path is contained in $G[V(C_j)\cup \{m_i,m_{i+1}\}]$. Otherwise for every $i\in \gamma^{-1}(j)$, let $P_i$ be the subpath in $P$ from $m_i$ to $m_{i+1}$. Let $m'_i, m'_{i+1}$ be the neighbours of $m_i,m_{i+1}$ in $P_i$ and Let $P'_i$ be  $m'_i$ to $m'_{i+1}$ subpath in $P_i$. We have that $P'_i$ is contained in $C_j$ and satisfies function $g$. Let $\cP=\{ P'_i|i\in [0,|M|], V(P'_i)\subseteq V(G-S)\}$. By the construction of $\cF_i$, we have that $\cP\cap \cF_i=P'_i$, that is  $|\cP\cap \cF_i|=1$ due to Item 5 of Reduction Rule \ref{rr:dmesp1}.   As $P$ realizes $\skel$, and for every vertex $v\in B$, $d_G(v,m_i)>\ell$, for all $i\in [0,|M|+1]$ and $d_G(v,u)>\ell-f(u)$, for all $u\in X$. Therefore, there exists a path in $\cP$ which contains a vertex $v'$ such that $d_G(v,v')\leq \ell$.  This implies that $\cP$ is a solution to $(F,B,\ell, \cF)$ of \amesp, and hence $(F,B,\ell, \cF)$ is a yes instance of \amesp. 
	
In the backward direction, suppose that $(F,B,\ell, \cF)$ is a yes instance of \amesp and let $\cP$ be its solution. Let $P'_i=\cF_i\cap \cP$ and $m'_i,m'_{i+1}$ be its end vertices such that $m_im'_i,m'_{i+1}m_{i+1} \in E(G)$. Let $P_i$ be the $m_i$ to $m_{i+1}$ path containing edges  $m_im'_i,{m'}_{i+1}m_{i+1}$ and path $P'_i$. We construct a path $P$ by concatenating paths $P_i$ in $\cP$ if $\gamma(i)\neq 0$ and edges $m_im_{i+1}$ when $\gamma(i)=0$. By the construction of $\cF_i$'s, $P$ satisfies functions $g$, $f$, $\gamma$ and ordering $\pi$ of $M$. Observe that for every vertex $v\in V(G)\setminus B$,  $d_G(v,m_i)\leq\ell$, for some $i\in [0,|M|+1]$ or $d_G(v,u)\leq \ell-f(u)$, for some $u\in X$. Therefore every vertex in $V(G)\setminus B$ is covered by $P$. Clearly, $P$ covers every vertex in $B$. Since we add only shortest paths in $\cF$ and as Reduction Rule \ref{rr:sanitytest2} is not applicable, $P$ is also a shortest path in $G$. This implies that $P$ is a solution  to $(G,S,k,\ell,\enskel)$ of \probenskel and hence $(G,S,k,\ell,\enskel)$ is a yes instance of \probenskel
\qed
\end{proof}

 \end{sloppypar}
 
 We design a dynamic programming-based algorithm for the \amesp problem parameterized by 
 %$|\cF|=t\leq k$. 
 $k$. Since the number of trees 
 %in $\cF$ 
 is at most $2k+2$, and the number of families of feasible paths 
 is $|\cF|=t$, 
 we first guess the subset of families of feasible paths that comes from each tree in $\cF$ in $\OO(k^{t})$ time. 
  Now we are ready to work on a tree with its guessed family of feasible paths. 
%In Section \ref{sec:mso}, we give an MSO$_2$ formula for the \amesp problem and show that \mespfvs is FPT. 
%We are now ready to describe an   algorithm for \amesp. 
%Due to space constraints, 
We first present an overview of the algorithm and then present the algorithm in Section \ref{sec:annotated}. 
%and give an overview of the algorithm in the next section. 
%In Appendix \ref{sec:annotated}, we give an explicit algorithm for \amesp using dynamic programming paradigm running in time $\OO(\ell^2 \cdot 2^{\OO(t)} n^{\OO(1)})$. 
%Due to space constraints, we 

\begin{lemma}\label{lem:annotated}
	%Given an instance $(F,B,  \ell, \cF)$ of \amesp, 
	\amesp~can be solved  in time   $\OO(\ell^2 \cdot 2^{\OO(k\log k)} n^{\OO(1)})$ when $F$ is a forest with  $\OO(k)$ trees. 
 %\todo[]{changed the running time. }
	% where $t$ is the number of sets in $\cF$ and $n$ is the number of vertices in $T$.
\end{lemma}

\subsection{Overview of the Algorithm for \amesp}\label{sec-path}

 %We first guess the Notice that the input to \amesp is a forest, and 

 %hence we can independently work with each connected component of the forest. Therefore, from now onwards, we assume that the input is restricted to a tree.  
 Consider an instance $(T,B,\ell,\cF =\{\cF_1,\cF_2,\dots,\cF_t\})$ of \amesp problem where $T$ is a tree, $B\subseteq V(T)$, and $\ell\in \mathbb{N}$ 
and $\cF$ is a disjoint family of feasible paths. The aim is to find a set $\cal P$ of $t$ paths such that for each $\cF_i$, $i\in [t]$, $|\mathcal{P}\cap \cF_i|=1$ and for every vertex $v\in B$, there exists a path $P\in {\cal P}$ and a vertex $u\in V(P)$, such that $d_T(u,v)\leq \ell$. 

For a vertex $v\in V(T)$, the bottom-up dynamic programming algorithm considers subproblems for each child $w$ of $v$ which are processed 
%We process its children 
from left to right. 
%(in index ordering) and at $i$th child, we consider  the subtree $T_{v,i}$. Before we define an entry in the table, we give the definition of the variables used for an entry. 
%Consider $v\in V(T)$ and $i\in [deg(v)]$. 
To compute a partial solution at the subtree rooted at a child of $v$, 
%$T_{v,i}$, 
we distinguish whether 
there exists a path containing $v$ that belongs to ${\cal P}$ or not.
%$v$ is contained in a path from ${\cal P}$.  
%a path in the solution. 
%If it is contained in a path in the solution, then we also guess the endpoints of the path. 
For this purpose, we define a variable that captures a path containing $v$ in ${\cal P}$. 
If there exists such a path, we guess the region where the endpoints of the path belong, which includes the cases that the 
%, which includes the cases where 
path contains: (i) only the vertex $v$, (ii) the parent of $v$ and one of its endpoints belongs to the subtree rooted at $w$ or $v$'s child that is to the left of $w$ or 
%$w$ or 
$v$'s child that is to the right of $w$,   
(iii) both its endpoints belong to the subtrees of the children which are to the left or the right of $w$, and 
(iv) one of the endpoints belongs to the subtree rooted at $w$ while the other belongs to the subtree of the child to the left or the right of $w$. An illustration of these cases is given in Fig. \ref{fig:dp}. 

At each node $v$, we store the distance of the nearest vertex (say $w'$) in the subtree of $v$, that is, on a path in $\cal P$, from $v$. 
We store this with the hope that $w'$ can cover vertices of $B$ that come in the future. 
%the future vertices can be covered by $w'$. 
In addition, we also store the farthest vertex (say $w''$) in the subtree of $v$ that is not covered by any chosen paths of ${\cal P}$ in the subtree. Again, we store this with the hope that $w''\in B$ can be covered by a future vertex, and the current solution leads to a solution overall.

At each node $v$, we capture the existence of the following: 
%of a state where we say a particular state is feasible if it satisfies the following 
there exists a set of $t'\leq t$ paths $Y$, one from each $\cF_i$, 
that either includes $v$ or not on a path from $Y$ 
%with the inclusion or exclusion of $v$ 
in ${\cal P}$ satisfying the distances of the nearest vertex $w'$ and the farthest vertex $w''$ (from $v$) 
that are on $Y$ and already covered and not yet covered by $Y$, respectively. To conclude the existence of a colorful path cover at the root node, we check for the existence of an entry that consists of a set $Y$ of $t$ paths, one from each $\cF_i$, 
and all the farthest distance of an uncovered vertex is zero. 

\subsection{Algorithm for \amesp}\label{sec:annotated}
%!TEX root = main.tex

%\section{Algorithm for \amesp}\label{sec:annotated}

In this section, we design a dynamic programming-based FPT algorithm for the \amesp problem parameterized by 
%$|\cF|=t\leq k$. 
$k$. 
Since the number of trees 
 %in $\cF$ 
 is at most $2k+2$, and the number of families of feasible paths 
 is $|\cF|=t$, 
 we first guess the subset of families of feasible paths that comes from each tree in $\cF$ in $\OO(k^{t})$ time.  
 Now we are ready to work on a tree with its guessed family of feasible paths. 
 %Notice that the input to \amesp is a forest, and hence we can independently work with each connected component of the forest. Therefore, from now onwards, we assume that the input is restricted to a tree.  
%In particular,  we prove \cref{thm:annotated}. 
Consider an instance $(T,B,\ell,\cF =\{\cF_1,\cF_2,\dots \cF_t\})$ of \amesp problem where $T$ is a tree, $B\subseteq V(T)$, and $\ell\in \mathbb{N}$ 
and $\cF$ is a disjoint family of feasible paths. 
%satisfying the following: 
The aim is to find a set $\cal P$ of $t$ paths such that for each $\cF_i$, $i\in [t]$, $|\mathcal{P}\cap \cF_i|=1$ and for every vertex $v\in B$, there exists a path $P\in {\cal P}$ and a vertex $u\in V(P)$, such that $d_T(u,v)\leq \ell$.

%\subsection{Algorithm for \annotatedmesp}\label{sec:annotated}
%
%In this section, we design a dynamic programming-based FPT algorithm for the \amesp problem. In particular we prove \cref{lem:annotated}.
%Consider an instance $(T,B,\ell,\cF=\{\cF_1,\cF_2,\dots \cF_t\})$ of \amesp problem. Recall that $T$ is a tree, $B\subseteq V(T)$, $\ell\in \mathbb{N}$, and $\cF$ satisfy the following: (1) Each $\cF_i$, $i\in [t]$ is a set of paths in $T$, and (2) $V(\cF_i)\cap V(\cF_j)=\emptyset$, for all $i,j\in [t], i\neq j$. The aim is to find a set $\cal P$ of $t$ paths such that for each $\cF_i$, $i\in [t]$, $\mathcal{P}\cap \cF_i=1$ and for every vertex $v\in B$, there exists a path $P\in {\cal P}$ and a vertex $u\in V(P)$, such that $d_T(u,v)\leq \ell$. 

% (3) If there exists two paths $P_1,P_2\in \cF_i$, $i\in [t]$, with endpoints $u_1,v_1$ and $u_2,v_2$, respectively, such that $V(P_1)\cap V(P_2)\neq \emptyset$, then there exists a vertex $v\in V(P_1)\cap V(P_2)$ such that $P_1$ and $P_2$ are paths that passes through $v$ and one of the following holds: 
%(i) Their endpoints are distinct, i.e., $u_1\neq u_2, u_1\neq v_2, v_1\neq u_2, v_1\neq v_2$ and 
%
%$V(P_1)\cap V(P_2)= \{v\}$.
%
%(ii) One of the endpoint is equal, for example $u_1=u_2$ and $v_1\neq v_2$ (other cases are analogous), in this case let $P'$ be the $v$ to $u_1$ subpath in $T$, then $V(P_1)\cap V(P_2)= V(P')$. 

%We design a dynamic programming algorithm. Intuition.

First, we give a description of our algorithm. The algorithm starts by arbitrarily rooting tree $T$ at a vertex $r\in V(T)$. 
%From now onwards, we assume that $T$ is rooted at $r$. 
In the following, we state some notations used in the algorithm. For a vertex $v\in V(T)$, we denote the number of children of $v$ in $T$ by $deg(v)$. By $T_v$, we denote the subtree of $T$ rooted at $v$. 
Also, the $i$th child of $v$ is denoted by $\sch^v_i$. 
For a vertex $v\in T$, and  $0 \leq i\leq deg(v)$, we define $T_{v,i}$ as the subtree of $T$ containing vertex $v$ and subtrees rooted at its first $i$ children (in order of index). Also,  $T_{\sch^v_j}$ is the 
subtree of $T$ containing vertex $v$ and subtree rooted at its  $i$th child. 
Recall that every vertex $v\in V(T)$ can be contained in at most one $\cF_i$, $i\in[t]$. For each vertex $v\in V(T)$, we assign a value $F(v)\in [0,t]$ to $v$ as follows: $F(v)=i$, if $v\in V(\cF_i)$, $0$ otherwise (when no path in any $\cF_i, i\in [t]$ contain $v$). 
%By $\sch^v_i$, we denote $i$th child of $v$. 
Recall that   a path $P$ {\em covers} a vertex $v$ if there exists a vertex $u\in V(P)$ such that $d_{G}(v,u)\leq \ell$. For a vertex $v\in V(T)$, the dynamic programming algorithm considers subproblems for each child of $v$. We process its children from left to right (in index ordering) and at $i$th child, we consider  the subtree $T_{v,i}$. Before we define an entry in the table, we give the definition of the variables used for an entry. 

Consider $v\in V(T)$ and $i\in [deg(v)]$. To compute a partial solution at subtree $T_{v,i}$, we must distinguish whether $v$ is contained in a path in the solution. If it is contained in a path in the solution, then we also guess the endpoints of the path. For this purpose, we define a variable $\stype$ which can take a value from the set $\{1, 2, \dots, 11\}$, where each value represents a different case defined as follows: 
\begin{itemize}
	\item\textbf{tp=1:} Paths containing $v$ are not in the solution. 
	\item\textbf{tp=2:} The vertex $v$ itself is a path that is contained in the solution. 
	
	%For the remaining cases, we have that a path containing $v$ is in the solution with the following added conditions: 
	
	%For $\stype\in \{3,4,5,6\}$, cases {\bf (3),(4),(5)} and {\bf (6)} we have condition that path also contains parent of $v$: with the following added conditions:
	\item\textbf{tp} {\boldmath$\in \{3, 4, 5, 6\}$}: For all these cases, we have the parent of $v$ in the solution, in addition to satisfying their respective properties.   
	\begin{itemize}
		\item\textbf{tp=3:} One of its endpoints is $v$ itself. 
		\item\textbf{tp=4:} One of its endpoints is in $T_{v,i-1}$ ($\neq v$). 
		\item\textbf{tp=5:} One of its endpoints is in $T_{\sch^v_i}$ ($\neq v$). 
		\item\textbf{tp=6:} One of its endpoints is in $T_{\sch^v_j}$ for some $j>i$ ($\neq v$). 
	\end{itemize}
	\item\textbf{tp=7:} Both of its endpoints are in $T_{v,i-1}$ (at least one endpoint is not equal to $v$).
	\item\textbf{tp=8:} One of its endpoints is in $T_{v,i-1}$, and the other 
	%other endpoint is 
	in $T_{\sch^v_i}$ ($\neq v$).  
	\item\textbf{tp=9:} One of its endpoints is in $T_{v,i-1}$, and the other 
	%endpoint is 
	in $T_{\sch^v_j}$ for some $j>i$ ($\neq v$). 
	\item\textbf{tp=10:} One of its endpoints is in $T_{\sch^v_i}$, and the 
	other 
	%endpoint is 
	in $T_{\sch^v_j}$ for some $j>i$ ($\neq v$). 
	\item\textbf{tp=11:}   One of its endpoints is in $T_{\sch^v_j}$ for some $j>i$, and 
	the other 
	%endpoint is 
	in $T_{\sch^v_{j'}}$ for some $j'>i$ ($\neq v$). 
	
	%\item $\bf 10$ represents that a path containing $v$ is in the solution, and both of its endpoints are in $T_{\sch^v_i}$ (at least one endpoint is not equal to $v$).

\end{itemize}

\begin{figure}[ht!]
	\centering
	\includegraphics[width=.8\textwidth]{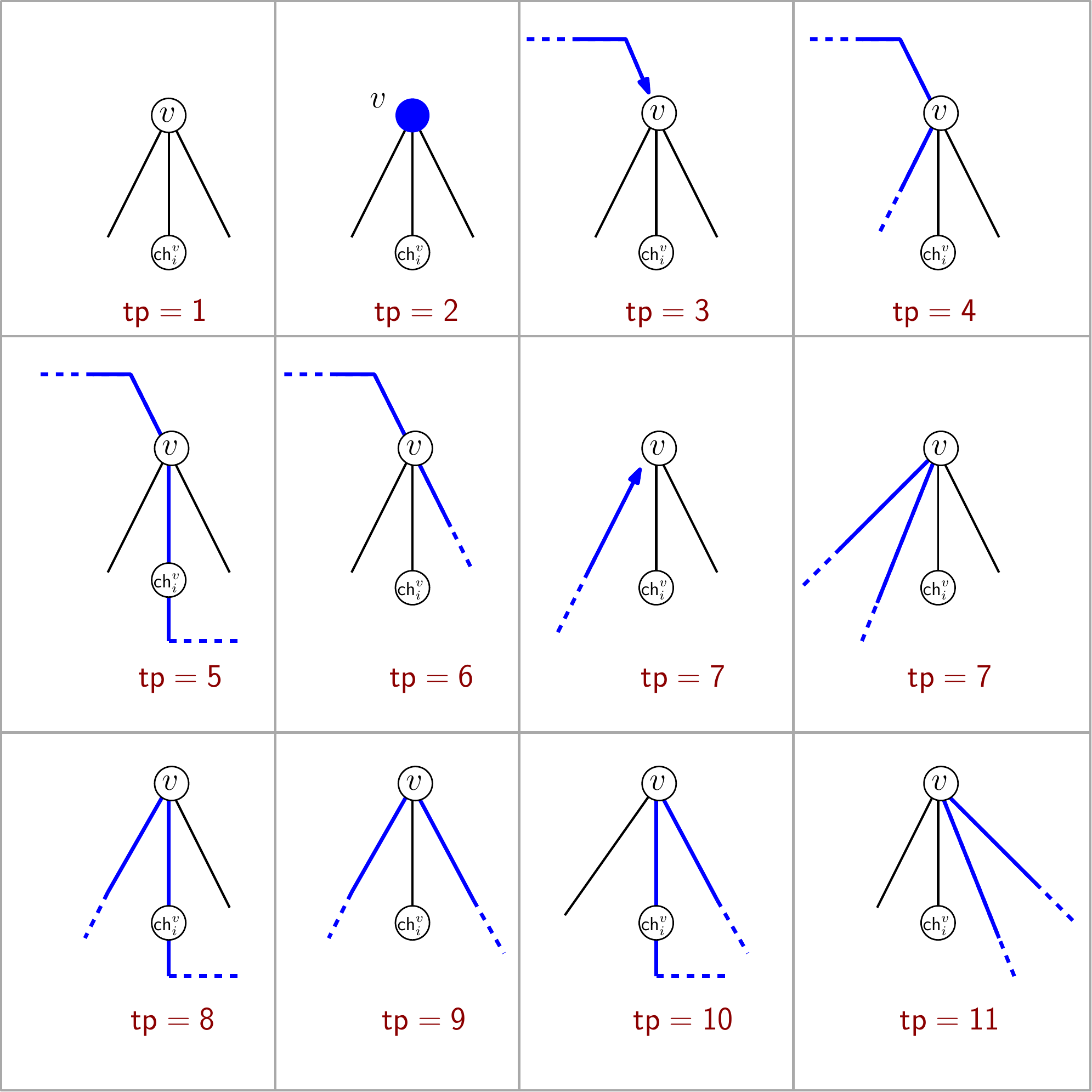}
	\caption{Illustration of the $\stype$ cases based on its value. The ``blue'' color indicates the path in $\cF$ that is a part of the solution path.}
	\label{fig:dp}
\end{figure}

\begin{sloppypar}
	We have a variable $(\ssign)\srem$. Here $\ssign$ represents sign of $\srem$, which can be either $+$ or $-$, and $(\ssign)\srem$ can take a value from the set $\{\sM\ell-1,\dots, \sM1, \sM0,\sP0,\sP1,\dots, \sP \ell\}$. The ``$-$'' sign represents that there are uncovered vertices in $V(T_{v,i})\cap B$ which needs to be covered in future, and ``$+$'' sign represents that all the vertices in $V(T_{v,i})\cap B$ are covered and the partial solution may cover more vertices outside $T_{v,i}$. The value $\sM\srem$ represents that the 
	maximum distance of a vertex $u\in B$ from $v$ in subtree $T_{v,i}$ such that $u$ is not covered by the partial solution is $\srem$. It means that vertices up to distance $\srem$ from $v$ in $V(T_{v,i})\cap B$ needs to be covered by some paths in future. We consider values only till $\ell-1$ distance, as it is a trivial observation that any vertex at distance at least $\ell$ from $v$ in $T_{v,i}$ cannot be covered by any path that is not in $T_{v,i}$. 
	%The value $\sP\srem$ represents the minimum value such that there exists a vertex \textcolor{red}{$u\in B$} in subtree $T_{v,i}$ such that distance from $u$ to $v$ is $\ell-\srem$ and $u$ is contained in a path in the solution. 
	The value $\sP\srem$ represents that the minimum distance of a vertex $u$ in subtree $T_{v,i}$ ($u$ need not belong to $B$) 
	such that distance from $u$ to $v$ is $\ell-\srem$ and $u$ is contained in a path in the solution. 
	It means that vertices in $B$ which are at a distance at most $\srem$ from $v$ (at most $\ell$ from $u$) in $T$ can be covered by a path containing $u$ in the solution.

	In the following definition, we state the entries in our dynamic programming routine.
	
	% \begin{definition} For each $v\in V(T)$, $d\in deg(v)$, $Y\subseteq [t]$,$\stype\in [12]$, ${\sf ext}\in [\ell], {\srem}\in [-1,\ell-1]$ we store an entry $c[v,d,Y,{\stype}, {\sf ext}, {\srem}]$ in our table which stores $\strue$ if and only if the following holds:
	
	\begin{definition} For each $v\in V(T)$, $d\in [deg(v)]$, $Y\subseteq [t]$,$\stype\in \{1, 2,\dots, 11\}, (\ssign)\srem \in\{\sM\ell-1,\dots, \sM1, \sM0,\sP0,\sP1,\dots, \sP \ell\}$ we create an entry $c[v,d,Y,{\stype},(\ssign) {\srem}]$ in our table which stores $\strue$ if and only if the following holds:
		\begin{itemize}
			\item  There exists a set $\cP_{v,d}$ of $|Y|$ paths such that for each $j\in Y$, $|\cP_{v,d}\cap \cF_j|=1$, if $F(v)\in Y$, then path in $\cP_{v,d}\cap \cF_{F(v)}$ should satisfy definition of $\stype$, and all paths in $\cP_{v,d}$ except maybe path containing $v$ should be contained in the subtree $T_{v,d}$.
			\item If $\ssign$ is $-$, then there exists a vertex $u\in B$ at distance $\srem$  from $v$ in $T_{v,d}$ which is not covered by any paths in $\cP_{v,d}$ but can be covered in future (is not at distance greater than $\ell$ from $v$). For every vertex $u\in B$ at distance at least ${\srem} +1$ from $v$ in subtree $T_{v,d}$ there exists a path in $\cP_{v,d}$ which covers $u$.
			\item If $\ssign$ is $+$, 
			then there exists a vertex $u$,  in $\cP_{v,d}$, 
			at distance $\ell - \srem$  from $v$ in $T_{v,d}$. 
			%which is not covered by any paths in $\cP_{v,d}$
			And $u$ can cover vertices (in future) that are at distance at most $\srem$ from $v$. 
			All other vertices in $\cP_{v,d}$ are at distance at least $\ell-\srem$ from $v$. 
			Further all vertices in $T_{v,d}\cap B$ are covered by $\cP_{v,d}$. 
			%for every vertex $u\in B$ which is either in $T_{v,d}$, or not in $T_{v,d}$ and at a distance at most ${\srem}$ from $v$ there exists a path in $\cP_{v,d}$ which covers $u$. 
		\end{itemize}
		
		Otherwise $c[v,d,Y,{\stype}, {\srem}]$ stores $\sfalse$.
		
	\end{definition}

	Observe that $(T,B,\ell,\cF=\{\cF_1,\cF_2,\dots,  \cF_t\})$ is a yes instance of \amesp if and only if there exists an entry $c[r,deg(r),\{1, 2, \dots, t\}, \stype^*,\sP\srem^*]$, $\stype^*\in \{1,2,7,8\}$ and  $\srem^* \in [0,\ell]$, which is set to $\strue$.

	Next, we give (recursive) formulas for the computation of entries in our dynamic programming table. Consider $v\in V(T),d\in [deg(v)],Y\subseteq [t],\stype\in [11],  (\ssign)\srem \in \{\sM\ell-1,\dots, \sM1, \sM0,\sP0,\sP1,\dots, \sP \ell\}$, we compute the entry $c[v,d,Y,{\stype}, (\ssign){\srem}]$ based on the following cases:
	
	\noindent {\bf Leaf vertex:} $v$ is a leaf vertex. 
	Set $c[v,d,Y,{\stype}, (\ssign){\srem}]=\strue$, if one of the following holds: (a1) $\stype=1, (\ssign)\srem=\sM 0$, $v\in B$, 
	or 
	(a2) $\stype\in \{2,3\} , (\ssign)\srem=\sP\ell$. 
	Otherwise, set the entry to $\sfalse$. Correctness follows trivially from the definition of entry.
	
	\noindent {\bf Non-leaf vertex:} $v$ is not a leaf vertex. We consider following cases depending on $\stype$ variable:
	
	\noindent {\bf Case $\stype=1$:} We further consider following cases.
	
	\textbf{(1)} If $\ssign=-$, then set $c[v,d,Y,1,(-){\srem}]=\strue$ if there exists a partition $Y_1\uplus Y_2$ of $Y$ such that one of the following holds:
	%\begin{itemize}
	%\item 
	(a1) $c[v,d-1,Y_1,1, (\ssign){\srem}_1]=\strue$, $(\ssign)\srem_1\in \{\sM\srem,\dots, \sP\srem-1\}$ and
	(a2) $c[\sch^v_d,deg(\sch^v_d),Y_2,1,\sM\srem-1]=\strue$,  
	%\item 
	(b1) $c[v,d-1,Y_1,1, \sM{\srem}]=\strue$, and
	(b2) $c[\sch^v_d,deg(\sch^v_d),Y_2,1,(\ssign)\srem_2]=\strue$, $(\ssign)\srem_2\in \{\sM\srem-1,\dots, \sP\srem\}$. 
	%\end{itemize}

	\textbf{(1.1)} If $\ssign=+$, then set $c[v,d,Y,1,(+)\ell -1]=\strue$ if there exists a partition $Y_1\uplus Y_2$ of $Y$ such that one of the following holds: (a1) $c[v,d-1,Y_1,1, (\ssign)\srem_1]=\strue$,   $(\ssign)\srem_1\in \{\sM\ell-1,\dots, \sP\ell-1\}$ and
	(a2) $c[\sch^v_d,deg(\sch^v_d),Y_2,\stype_2,\sP\ell]=\strue$, $\stype_2\in \{2,7,8\}$, or, 
	(b1) $c[v,d-1,Y_1,1, \sP\ell-1]=\strue$, 
	and
	(b2) $c[\sch^v_d,deg(\sch^v_d),Y_2,\stype_2,(\ssign)\srem_2]=\strue$, $\stype_2\in \{1, 2,7,8\}$, $(\ssign)\srem_2\in \{\sM\ell-2,\dots, \sP\ell\}$. 
	
	\textbf{(1.2)} If $\ssign=+$ and $\srem<\ell-1$, then set $c[v,d,Y,1,(+){\srem}]=\strue$ if there exists a partition $Y_1\uplus Y_2$ of $Y$ such that one of the following holds: (a1) $c[v,d-1,Y_1,1, (\ssign){\srem}_1]=\strue$,  $(\ssign)\srem_1\in \{\sM\srem,\dots, \sP\srem\}$ and
	(a2) $c[\sch^v_d,deg(\sch^v_d),Y_2,1,\sP\srem+1]=\strue$,  or
	(b1) $c[v,d-1,Y_1,1, \sP{\srem}]=\strue$,  and
	(b2) $c[\sch^v_d,deg(\sch^v_d),Y_2,1,(\ssign)\srem_2]=\strue$,  $(\ssign)\srem_2\in \{\sM\srem,\dots, \sP\srem+1\}$.

	%(2) If $\ssign=+$, then set $c[v,d,Y,1,(+){\srem}]=\strue$ if there exists a partition $Y_1\uplus Y_2$ of $Y$ such that one of the following holds:
	% \begin{itemize}
	%\item (a1) $c[v,d-1,Y_1,{\stype}_1, (\ssign){\srem}_1]=\strue$,  $\stype_1\in \{1,2,7,8\}$, $(\ssign)\srem_1\in \{\sM\srem,\dots, \sP\srem+1\}$ and\\
	% (a2) $c[\sch^v_d,deg(\sch^v_d),Y_2,\stype_2,\sP\srem+1]=\strue$, $\stype_2\in \{1,2,7,8\}$.
	%  \item (b1) $c[v,d-1,Y_1,{\stype}_1, \sP{\srem}+1]=\strue$, $\stype_1\in \{1,2,7,8\}$ and\\
	% (b2) $c[\sch^v_d,deg(\sch^v_d),Y_2,{\stype}_2,(\ssign)\srem_2]=\strue$, $\stype_1\in \{1,2,7,8\}$, $(\ssign)\srem_2\in \{\sM\srem,\dots, \sP\srem+1\}$
	%  \end{itemize}

	\noindent {\bf Correctness of Case \stype=1 (1):} $\ssign=-$: Consider the case when (a1) and (a2) are true, then there exists a set of paths $\cP_1$ which set (a1) to true and there is a set of paths $\cP_2$ which set (a2) to true. We claim that $\cP=\cP_1\cup \cP_2$ is a solution to $c[v,d,Y,1,(-){\srem}]=\strue$. Observe that no path containing $v$ is in $\cP$, as $\stype=1$ in (a1). The condition (a2) implies that there exists an uncovered vertex $u\in T_{\sch^v_d}$ at distance $\srem$ from $v$. If $\sch^v_d\in \cP_2$, then $u$ will be  covered by $\sch^v_d$ and $u$ need not wait for future vertices to cover it. Hence $\stype=1$ in (a2). 
	%Also $\sch^v_d\notin \cP_2$. Else the vertex 
	%and (a2). 
	%$\stype_1,\stype_2\in \{1,2,7,8\}$. 
	Since $Y=Y_1\uplus Y_2$, we have that for each $i\in Y$, we have exactly one path in $\cP\cap \cF_i$. Next, observe that there exists a vertex $u\in B$ at distance $\srem$ from $v$ in $T_{\sch^v_d}$, which is not covered by paths in $\cP_1$ and $\cP_2$. Also $(\ssign)\srem_1\in \{\sM\srem,\dots, \sP\srem-1\}$, 
	$u$ is not covered by  $\cP_1$, 
	as any vertex in $\cP_1$ requires at least $\srem+1$ distance to reach $u$.  Observe that every vertex at distance at least $\srem+1$ from $v$ in $T_{v,d}$ is covered by paths in $\cP$. Analogous arguments follows for the case when (b1) and (b2) are true.

	\noindent\textbf{Correctness of Case $\stype=1$ (1.1):} $\ssign=+$: 
	Since $\stype=1$, $v\notin \cP$. 
	Hence, we do not consider the entry $c[v,d,Y,1,(+)\ell]$. 
	Consider the case when (a1) and (a2) are true, then there exists a set of paths $\cP_1$ which set (a1) to true and there is a set of paths $\cP_2$ which set (a2) to true. We claim that $\cP=\cP_1\cup \cP_2$ is a solution to $c[v,d,Y,1,(+)\ell-1]=\strue$. 
	As $\stype_2\in \{2,7,8\}$ and no path in $\cP_1$ contain $v$, 
	we have that no path containing $v$ is in $\cP$. 
	%Observe that no path containing $v$ is in $\cP$, as $\stype_1,\stype_2\in \{1,2,7,8\}$. 
	Since $Y=Y_1\uplus Y_2$, we have that for each $i\in Y$, we have exactly one path in $\cP\cap \cF_i$. Next, observe that every vertex in $T_{\sch^v_d}$ is covered by $\cP_2$ and since  $(\ssign)\srem_1\in \{\sM\ell-1,\dots, \sP\ell-1\}$, every vertex in $T_{v,d-1}$ is covered either by $\cP_1$ or $\cP_2$. By (a2) we also satisfy $\sP\ell$, that is vertices at distance at most $\ell$ is covered by $\cP_2$. Analogous arguments follows for the case when (b1) and (b2) are true.

	%{\em Correctness}: Consider the case when (a1) and (a2) are true, then there exists a set of paths $\cP_1$ which set (a1) to true and there is a set of paths $\cP_2$ which set (a2) to true. We claim that $\cP=\cP_1\cup \cP_2$ is a solution to $c[v,d,Y,1,(-){\srem}]=\strue$. 
	%Observe that no path containing $v$ is in $\cP$, as $\stype_1,\stype_2\in \{1,2,7,8\}$. since $Y=Y_1\uplus Y_2$, we have that for each $i\in Y$, we have exactly one path in $\cP\cap \cF_i$. Next, observe that there exists a vertex $u\in B$ at distance $\srem$ from $v$ in $T_{\sch^v_d}$, which is not covered by paths in $\cP_1$. Also $(\ssign)\srem_1\in \{\sM\srem-1,\dots, \sP\srem\}$, it is not covered by  $\cP_2$, as any vertex in $\cP_2$ requires at least $\srem+1$ distance to reach $u$. Observe that every vertex at distance at least $\srem+1$ from $v$ in $T_{v,d}$ is covered by paths in $\cP$. Analogous arguments follows for the case when (b1) and (b2) are true.

	\noindent\textbf{Correctness of $\stype=1$ (1.2):} $\ssign=+$: 
	Since $\srem<\ell-1$, none of the children of $v$ can be in $\cP$. Hence $\stype=1$ in both (a1) and (a2). 
	Consider the case when (a1) and (a2) are true, then there exists a set of paths $\cP_1$ which set (a1) to true and there is a set of paths $\cP_2$ which set (a2) to true. We claim that $\cP=\cP_1\cup \cP_2$ is a solution to $c[v,d,Y,1,(+){\srem}]=\strue$. 
	%Observe that no path containing $v$ is in $\cP$, as $\stype =1$ in (a1) and (a2). 
	%$\stype_1,\stype_2\in \{1,2,7,8\}$. 
	Since $Y=Y_1\uplus Y_2$, we have that for each $i\in Y$, we have exactly one path in $\cP\cap \cF_i$. Next, observe that every vertex in $T_{\sch^v_d}$ is covered by $\cP_2$ and since  $(\ssign)\srem_1\in \{\sM\srem,\dots, \sP\srem\}$, 
	every vertex in $T_{v,d-1}$ is covered either by $\cP_1$ or $\cP_2$. By (a2) we also satisfy $\sP\srem$, that is vertices at distance at most $\srem$ is covered by $\cP_2$. Analogous arguments follows for the case when (b1) and (b2) are true.

	%(2) If $\ssign=+$, then set $c[v,d,Y,1,(+){\srem}]=\strue$ if there exists a partition $Y_1\uplus Y_2$ of $Y$ such that one of the following holds:
	% \begin{itemize}
	%\item (a1) $c[v,d-1,Y_1,{\stype}_1, (\ssign){\srem}_1]=\strue$,  $\stype_1\in \{1,2,7,8\}$, $(\ssign)\srem_1\in \{\sM\srem,\dots, \sP\srem+1\}$ and\\
	% (a2) $c[\sch^v_d,deg(\sch^v_d),Y_2,\stype_2,\sP\srem+1]=\strue$, $\stype_2\in \{1,2,7,8\}$.
	%  \item (b1) $c[v,d-1,Y_1,{\stype}_1, \sP{\srem}+1]=\strue$, $\stype_1\in \{1,2,7,8\}$ and\\
	% (b2) $c[\sch^v_d,deg(\sch^v_d),Y_2,{\stype}_2,(\ssign)\srem_2]=\strue$, $\stype_1\in \{1,2,7,8\}$, $(\ssign)\srem_2\in \{\sM\srem,\dots, \sP\srem+1\}$
	%  \end{itemize}
	
	%{\em Correctness}: Consider the case when (a1) and (a2) are true, then there exists a set of paths $\cP_1$ which set (a1) to true and there is a set of paths $\cP_2$ which set (a2) to true. We claim that $\cP=\cP_1\cup \cP_2$ is a solution to $c[v,d,Y,1,(+){\srem}]=\strue$. Observe that no path containing $v$ is in $\cP$, as $\stype_1,\stype_2\in \{1,2,7,8\}$. since $Y=Y_1\uplus Y_2$, we have that for each $i\in Y$, we have exactly one path in $\cP\cap \cF_i$. Next, observe that every vertex in $T_{\sch^v_d}$ is covered by $\cP_2$ and since  $(\ssign)\srem_1\in \{\sM\srem,\dots, \sP\srem+1\}$, every vertex in $T_{v,d-1}$ is covered either by $\cP_1$ or $\cP_2$. By (a2) we also satisfy $\sP\srem$, that is vertices at distance at most $\srem$ is covered by $\cP_2$. Analogous arguments follows for the case when (b1) and (b2) are true.
	
	\noindent {\bf Case  $\stype=2$:} If $F(v)\notin Y$, or $\ssign=\sM$, or $\srem\neq \ell$, or $v$ is not a path in $\cF_{F(v)}$, then set $c[v,d,Y,2,(\ssign){\srem}]=\sfalse$. Correctness follows trivially. Otherwise we set $c[v,d,Y,2,(+)\ell]=\strue$ if there exists a partition $Y_1\uplus Y_2$ of $Y\setminus \{F(v)\}$ such that the following holds: (a1) $c[v,d-1,Y_1\cup \{F(v)\},2, \sP\ell]=\strue$ and (a2) $c[\sch^v_d,deg(\sch^v_d),Y_2,\stype_2,(\ssign)\srem_2]=\strue$,  $(\ssign)\srem_2\in \{\sM\ell-1,\dots, \sP\ell\}, \stype_2\in \{1,2,7,8\}$.
	
	%\noindent {\bf Case $\stype=2$:}  If $F(v)\notin Y$, or $\ssign=\sM$, or $\srem\neq \ell$, or $v$ is not a path in $\cF_{F(v)}$, then set $c[v,d,Y,2,(\ssign){\srem}]=\sfalse$. Correctness follows trivially. Otherwise we set $c[v,d,Y,2,(+)\ell]=\strue$ if there exists a partition $Y_1\uplus Y_2$ of 
	%$Y\setminus \{F(v)\}$ such that the following holds: (a1) $c[v,d-1,Y_1\cup \{F(v)\},2, (+)\ell]=\strue$ and (a2) $c[\sch^v_d,deg(\sch^v_d),Y_2,{\stype}_2,(\ssign)\srem_2]=\strue$,  $\stype_1, \stype_2\in \{1,2,7,8\}$, $(\ssign)\srem_1, (\ssign)\srem_2\in \{\sM\ell-1,\dots, \sP\ell\}$.

	\noindent {\bf Correctness of Case $\stype=2$:} Suppose that (a1) and (a2) are true, then there exists a set of paths $\cP_1$ which set (a1) to true and there is a set of paths $\cP_2$ which set (a2) to true. We claim that $\cP=\cP_1\cup \cP_2 \cup \{v\}$ is a solution to $c[v,d,Y,2,(+)\ell]=\strue$. %\textcolor{red}{Observe that any neighbor of $v$ cannot be there in the path $\mathcal{P}$, by observation }. Hence all the neighbors of $v$ are not in the solution. 
	%Thus $\sch^v_d$ also has $\stype=1$ in (a2). 
	%Observe that no path containing $v$ is in $\cP$, as $\stype_1,\stype_2\in \{1,2,7,8\}$. 
	Since $Y=Y_1\uplus Y_2\uplus \{F(v)\}$, we have that for each $i\in Y$, we have exactly one path in $\cP\cap \cF_i$. Next, observe that every vertex in $T_{\sch^v_d} \cup  T_{v,d-1}$ is covered by $\cP$ since  $(\ssign)\srem_2\in \{\sM\ell-1,\dots, \sP\ell\}$ and path containing $v$ is in $\cP$,  hence every vertex in $T_{v,d}$ is covered by $\cP$.

	%{\em Correctness}: Suppose that (a1) and (a2) are true, then there exists a set of paths $\cP_1$ which set (a1) to true and there is a set of paths $\cP_2$ which set (a2) to true. We claim that $\cP=\cP_1\cup \cP_2 \cup \{v\}$ is a solution to $c[v,d,Y,2,(+)\ell]=\strue$. Observe that no path containing $v$ is in $\cP$, as $\stype_1,\stype_2\in \{1,2,7,8\}$. since $Y=Y_1\uplus Y_2\uplus \{F(v)\}$, we have that for each $i\in Y$, we have exactly one path in $\cP\cap \cF_i$. Next, observe that every vertex in $T_{\sch^v_d} \cup  T_{v,d-1}$ is covered by $\cP$ since  $(\ssign)\srem_1\in \{\sM\srem,\dots, \sP\srem+1\}$ and path containing $v$ is in $\cP$,  hence every vertex in $T_{v,d}$ is covered by $\cP$.
	%
	\noindent {\bf Case $\stype=3$:}  If $F(v)\notin Y$, or $\ssign=\sM$, or $\srem\neq \ell$, or there is no path in $\cF_{F(v)}$ containing $v$, parent of $v$ and has one of its endpoints as $v$, then set $c[v,d,Y,3,(\ssign){\srem}]=\sfalse$. Correctness follows trivially. Otherwise we set $c[v,d,Y,3,(+)\ell]=\strue$ if there exists a partition $Y_1\uplus Y_2$ of $Y\setminus \{F(v)\}$ such that the following holds: (a1) $c[v,d-1,Y_1\cup \{F(v)\},3, (+)\ell]=\strue$ and (a2) $c[\sch^v_d,deg(\sch^v_d),Y_2,{\stype}_2,(\ssign)\srem_2]=\strue$,  $\stype_2\in \{1,2,7,8\}$, $(\ssign)\srem_2\in \{\sM\ell-1,\dots, \sP\ell\}$. The correctness can be argued similar to Case 2.
	
	%	\noindent {\bf Case $\stype=3$:}  If $F(v)\notin Y$, or $\ssign=\sM$, or $\srem\neq \ell$, or there is no path in $\cF_{F(v)}$ containing $v$, parent of $v$ and has one of its endpoint as $v$, then set $c[v,d,Y,3,(\ssign){\srem}]=\sfalse$. Correctness follows trivially. Otherwise we set $c[v,d,Y,3,(+)\ell]=\strue$ if there exists a partition $Y_1\uplus Y_2$ of $Y\setminus \{F(v)\}$ such that the following holds: (a1) $c[v,d-1,Y_1\cup \{F(v)\},3, (+)\ell]=\strue$ and (a2) $c[\sch^v_d,deg(\sch^v_d),Y_2,1,(\ssign)\srem_2]=\strue$, 
	%	$(\ssign)\srem_2\in \{\sM\ell-1,\dots, \sP\ell-1\}$.
	% The correctness can be argued similar to Case 2.

	%\noindent {\bf Case $\stype=4$:}  If $F(v)\notin Y$, or $\ssign=\sM$, or $\srem\neq \ell$, or there is no path in $\cF_{F(v)}$ containing $v$, parent of $v$ and has one of its endpoints in $T_{v,d-1}$ (endpoint is not $v$), then set $c[v,d,Y,4,(\ssign){\srem}]=\sfalse$. Correctness follows trivially. Otherwise we set $c[v,d,Y,4,(+)\ell]=\strue$ if there exists a partition $Y_1\uplus Y_2$ of $Y\setminus \{F(v)\}$ such that the following holds: (a1) $c[v,d-1,Y_1\cup \{F(v)\},\stype_1, (+)\ell]=\strue$, $\stype\in \{3,4\}$ and (a2) $c[\sch^v_d,deg(\sch^v_d),Y_2,{\stype}_2,(\ssign)\srem_2]=\strue$,  $\stype_2\in \{1,2,7,8\}$, $(\ssign)\srem_2\in \{\sM\ell-1,\dots, \sP\ell\}$. The correctness can be argued similar to Case 2.

	\noindent {\bf Case $\stype=4$:}  If $F(v)\notin Y$, or $\ssign=\sM$, or $\srem\neq \ell$, or there is no path in $\cF_{F(v)}$ containing $v$, parent of $v$ and has one of its endpoints in $T_{v,d-1}$ (endpoint is not $v$), then set $c[v,d,Y,4,(\ssign){\srem}]=\sfalse$. Correctness follows trivially. Otherwise we set $c[v,d,Y,4,(+)\ell]=\strue$ if there exists a partition $Y_1\uplus Y_2$ of $Y\setminus \{F(v)\}$ such that the following holds: (a1) $c[v,d-1,Y_1\cup \{F(v)\},4, (+)\ell]=\strue$, 
	and (a2) $c[\sch^v_d,deg(\sch^v_d),Y_2,\stype_2 ,(\ssign)\srem_2]=\strue$, $\stype_2\in \{1,2,7,8\}$, 
	$(\ssign)\srem_2\in \{\sM\ell-1,\dots, \sP\ell\}$. 
	The correctness can be argued similar to Case 2.

	\noindent {\bf Case $\stype=5$:}  If $F(v)\notin Y$, or $\ssign=\sM$, or $\srem\neq \ell$, or there is no path in $\cF_{F(v)}$ containing $v$, parent of $v$ and has one of its endpoints in $T_{\sch^v_d}$ (endpoint is not $v$), then set $c[v,d,Y,5,(\ssign){\srem}]=\sfalse$. Correctness follows trivially. Otherwise we set $c[v,d,Y,5,(+)\ell]=\strue$ if there exists a partition $Y_1\uplus Y_2$ of $Y\setminus \{F(v)\}$ such that the following holds: (a1) $c[v,d-1,Y_1\cup \{F(v)\},\stype_1, (+)\ell]=\strue$, $\stype_1=6$ and (a2) $c[\sch^v_d,deg(\sch^v_d),Y_2 \cup \{F(v)\},{\stype}_2,(+)\ell]=\strue$,  $\stype_2\in \{3,4,5\}$. 
	
	%\noindent {\bf Case  $\stype=5$:} If $F(v)\notin Y$, or $\ssign=\sM$, or $\srem\neq \ell$, or there is no path in $\cF_{F(v)}$ containing $v$, parent of $v$ and has one of its endpoints in $T_{\sch^v_d}$ (endpoint is not $v$), then set $c[v,d,Y,5,(\ssign){\srem}]=\sfalse$. Correctness follows trivially. Otherwise we set $c[v,d,Y,5,(+)\ell]=\strue$ if there exists a partition $Y_1\uplus Y_2$ of $Y\setminus \{F(v)\}$ such that the following holds: (a1) $c[v,d-1,Y_1\cup \{F(v)\},6, (+)\ell]=\strue$, and (a2) $c[\sch^v_d,deg(\sch^v_d),Y_2 \cup \{F(v)\},{\stype}_2,(+)\ell]=\strue$,  $\stype_2\in \{3,4,5\}$.

	\noindent\textbf{Correctness of the Case $\stype=5$:} Suppose that (a1) and (a2) are true, then there exists a set of paths $\cP_1$ which set (a1) to true and there is a set of paths $\cP_2$ which set (a2) to true. Let $P_1\in \cP_1 \subseteq \cF_F(v)$ and let $P_2\in \cP_1 \subseteq \cF_F(v)$. Observe that $\cP^*_1=(\cP_1\setminus \{P_1\})\cup \{P_2\}$ also sets (a1) to true. We claim that $\cP=\cP^*_1\cup \cP_2$ is also a solution to $c[v,d,Y,5,(+)\ell]=\strue$.  since $Y=Y_1\uplus Y_2\uplus \{F(v)\}$, and $\stype_1=6$ and $\stype_2\in \{3,4,5\}$, and exactly one path containing $v$ is in $\cP$,  we have that for each $i\in Y$, we have exactly one path in $\cP\cap \cF_i$. Next, observe that every vertex in $T_{\sch^v_d} \cup T_{v,d-1}$ is covered by $\cP$, hence every vertex in $T_{v,d}$ is covered by $\cP$.

	\noindent {\bf Case $\stype=6$:}  If $F(v)\notin Y$, or $d=deg(v)$, or $\ssign=\sM$, or $\srem\neq \ell$, or there is no path in $\cF_{F(v)}$ containing $v$, parent of $v$ and has one of its endpoints in $T_{\sch^v_{d'}}, d'>d$ (endpoint is not $v$), then set $c[v,d,Y,6,(\ssign){\srem}]=\sfalse$. Correctness follows trivially. Otherwise we set $c[v,d,Y,6,(+){\srem}]=\strue$ if there exists a partition $Y_1\uplus Y_2$ of $Y\setminus \{F(v)\}$ such that the following holds: (a1) $c[v,d-1,Y_1\cup \{F(v)\},\stype_1, (+)\ell]=\strue$, and (a2) $c[\sch^v_i,deg(\sch^v_i),Y_2,{\stype}_2,(+)\ell]=\strue$, where $\stype_1,\stype_2=6$. The correctness can be argued similar to Case 5.

	\noindent {\bf Case $\stype=7$:}  If $F(v)\notin Y$, or $\ssign=\sM$, or $\srem\neq \ell$, or there is no path in $\cF_{F(v)}$ containing $v$, and both of its endpoints in $T_{d,i-1}$ (at least one endpoint is not $v$), then set $c[v,d,Y,7,(\ssign){\srem}]=\sfalse$. Correctness follows trivially. Otherwise we set $c[v,d,Y,7,(+){\ell}]=\strue$ if there exists a partition $Y_1\uplus Y_2$ of $Y\setminus \{F(v)\}$ such that the following holds: (a1) $c[v,d-1,Y_1\cup \{F(v)\},\stype_1, (+)\ell]=\strue$, $\stype_1\in \{7,8\}$ and (a2) $c[\sch^v_d,deg(\sch^v_d),Y_2,{\stype}_2,(\ssign)\srem_2]=\strue$,  $\stype_2 \in \{1,2,7,8\}$, $(\ssign)\srem_2\in \{\sM\ell-1,\dots, \sP\ell\}$. The correctness can be argued similar to Case 5.

	%\noindent {\bf Case $\stype=7$:}  If $F(v)\notin Y$, or $\ssign=\sM$, or $\srem\neq \ell$, or there is no path in $\cF_{F(v)}$ containing $v$, and both of its endpoints in $T_{d,i-1}$ (at least one endpoint is not $v$), then set $c[v,d,Y,7,(\ssign){\srem}]=\sfalse$. Correctness follows trivially. Otherwise we set $c[v,d,Y,7,(+){\ell}]=\strue$ if there exists a partition $Y_1\uplus Y_2$ of $Y\setminus \{F(v)\}$ such that the following holds: (a1) $c[v,d-1,Y_1\cup \{F(v)\},\stype_1, (+)\ell]=\strue$, $\stype_1\in \{7,8,9\}$ and (a2) $c[\sch^v_d,deg(\sch^v_d),Y_2,{\stype}_2,(\ssign)\srem_2]=\strue$,  $\stype_2 \in \{1,2,7,8\}$, $(\ssign)\srem_2\in \{\sM\ell-1,\dots, \sP\ell\}$. The correctness can be argued similar to Case 5.

	\noindent {\bf Case $\stype=8$:}  If $F(v)\notin Y$, or $\ssign=\sM$, or $\srem\neq \ell$, or there is no path in $\cF_{F(v)}$ containing $v$, and one of its endpoints in $T_{d,i-1}$ and other endpoint is in $T_{\sch^v_d}$ (endpoint is not $v$), then set $c[v,d,Y,8,(\ssign){\srem}]=\sfalse$. Correctness follows trivially. Otherwise we set $c[v,d,Y,8,(+){\ell}]=\strue$ if there exists a partition $Y_1\uplus Y_2$ of $Y\setminus \{F(v)\}$ such that the following holds: (a1) $c[v,d-1,Y_1\cup \{F(v)\},\stype_1, (+)\ell]=\strue$, $\stype_1=9$ and (a2) $c[\sch^v_i,deg(\sch^v_i),Y_2\cup \{F(v)\},{\stype}_2,(+)\ell]=\strue$,  $\stype_2\in \{3,4,5\}$, $(\ssign)\srem_2\in \{\sM\ell-1,\dots, \sP\ell\}$. The correctness can be argued similar to Case 5.
	
	\noindent {\bf Case $\stype=9$:}  If $F(v)\notin Y$, or $d=deg(v)$ or $\ssign=\sM$, or $\srem\neq \ell$, or there is no path in $\cF_{F(v)}$ containing $v$, and one of its endpoints in $T_{d,i-1}$ and other endpoint is in $T_{\sch^v_{d'}}$ for some $d'>d$ (endpoint is not $v$), then set $c[v,d,Y,9,(\ssign){\srem}]=\sfalse$. Correctness follows trivially. Otherwise we set $c[v,d,Y,9,(+){\ell}]=\strue$ if there exists a partition $Y_1\uplus Y_2$ of $Y\setminus \{F(v)\}$ such that the following holds: (a1) $c[v,d-1,Y_1\cup \{F(v)\},\stype_1, (+)\ell]=\strue$, $\stype_1=9$ and (a2) $c[\sch^v_d,deg(\sch^v_d),Y_2,{\stype}_2,(+)\ell]=\strue$,  $\stype_2\in \{1,2,7,8\}$, $(\ssign)\srem_2\in \{\sM\ell-1,\dots, \sP\ell\}$. The correctness can be argued similar to Case 5.

	% \noindent {\bf Case 10} $\stype=10$ : If $F(v)\notin Y$, or $\ssign=\sM$, or $\srem\neq \ell$, or there is no path in $\cF_{F(v)}$ containing $v$, and  both of its endpoints are in $T_{\sch^v_d}$ (at least one endpoint is not equal to $v$), then set $c[v,d,Y,10,(\ssign){\srem}]=\sfalse$. Correctness follows trivially. Otherwise we set $c[v,d,Y,10,(+){\ell}]=\strue$ if there exists a partition $Y_1\uplus Y_2$ of $Y\setminus \{F(v)\}$ such that the following holds: (a1) $c[v,d-1,Y_1\cup \{F(v)\},\stype_1, (+)\ell]=\strue$, $\stype_1=9$ and (a2) $c[\sch^v_d,deg(\sch^v_d),Y_2\cup \{F(v)\},{\stype}_2,(+)\ell]=\strue$,  $\stype_2\in \{3,4,5\}$, $(\ssign)\srem_2\in \{\sM\ell-1,\dots, \sP\ell\}$. The correctness can be argued similar to Case 5.
	
	\noindent {\bf Case $\stype=10$:}  If $F(v)\notin Y$, or $d=deg(v)$, or $\ssign=\sM$, or $\srem\neq \ell$, or there is no path in $\cF_{F(v)}$ containing $v$, and one of its endpoints is in $T_{\sch^v_d}$, and other endpoint is in $T_{\sch^v_{d'}}$ for some $d'>d$ (endpoint is not $v$), then set $c[v,d,Y,10,(\ssign){\srem}]=\sfalse$. Correctness follows trivially. Otherwise we set $c[v,d,Y,10,(+){\ell}]=\strue$ if there exists a partition $Y_1\uplus Y_2$ of $Y\setminus \{F(v)\}$ such that the following holds: (a1) $c[v,d-1,Y_1\cup \{F(v)\},\stype_1, (+)\ell]=\strue$, $\stype_1=11$ and (a2) $c[\sch^v_d,deg(\sch^v_d),Y_2\cup \{F(v)\},{\stype}_2,(+)\ell]=\strue$,  $\stype_2=10$. The correctness can be argued similar to Case 5.
	
	\noindent {\bf Case $\stype=11$:}  If $F(v)\notin Y$, or $d=deg(v)$, or $\ssign=\sM$, or $\srem\neq \ell$, or there is no path in $\cF_{F(v)}$ containing $v$, one its endpoints is in $T_{\sch^v_{d'}}$ for some $d'>d$, and other endpoint is in $T_{\sch^v_{d''}}$ for some $d''>d$ (endpoint is not $v$), then set $c[v,d,Y,11,(\ssign){\srem}]=\sfalse$. Correctness follows trivially. Otherwise we set $c[v,d,Y,11,(+){\ell}]=\strue$ if there exists a partition $Y_1\uplus Y_2$ of $Y\setminus \{F(v)\}$ such that the following holds:(a1) $c[v,d-1,Y_1\cup \{F(v)\},\stype_1, (+)\ell]=\strue$, $\stype_1=11$ and (a2) $c[\sch^v_d,deg(\sch^v_d),Y_2,{\stype}_2,(+)\ell]=\strue$,  $\stype_2\in \{1,2,7,8\}$, $(\ssign)\srem_2\in \{\sM\ell-1,\dots, \sP\ell\}$. The correctness can be argued similar to Case 5.
	
\end{sloppypar}

This completes the description of the (recursive) formulas and their correctness for computing all entries of the dynamic programming table. The correctness of the algorithm follows from the correctness of the (recursive) formulas, and the fact $(T,B,\ell,\cF=\{\cF_1,\cF_2,\dots, \cF_t\})$ is a yes instance of \amesp if and only there exists an entry $c[r,deg(r),\stype^*,\sP\srem^*]$, $\stype^*\in \{1,2,7,8\}$ and  $\srem^* \in [0,\ell]$, which is set to $\strue$. Next, we analyse the running time of our algorithm. Recall that $|\cF|\leq t$ and for each $i\in [t]$ we have that  $|\cF_i|\leq n^2$. 
Observe that each of our table entries can be computed in time $\OO(\ell^2 \cdot 2^{2t}n^{\OO(1)})$ time. The number of entries is bounded by $\OO(\ell\cdot  2^t n^{\OO(1)})$. 

Since we guessed the subset of families of feasible paths that come from each tree in $\cF$ in $\OO(k^{t})$ time, 
the running time of our algorithm is bounded by $\OO(\ell^2 \cdot 2^{\OO(t\log k)} n^{\OO(1)})$. Since 
$t\leq k$, we get the desired running time. 

%\vspace{-0.5cm}

\subsection{Proof of Theorem~\ref{thm:fvs}}
\mespfvs and \disjointmesp are FPT-equivalent from Lemma \ref{lem:mespfvscorr}. Observation~\ref{lem:skelnum} upper bounds the number of skeletons by $2^{\OO(k(\log k + \log \ell))}n^2$. 
Then, 
%in \cref{sec:covering}, 
we show that \dmesp and \probenskel are FPT-equivalent and 
%and we showed in Section~\ref{sec:covering} that 
for each skeleton we have at most $2^{\OO(k\log k)}$ enriched skeletons. 
Finally, given an instance of \probenskel, we construct an instance of \amesp in polynomial time. The
\amesp~problem can be solved in $\OO(\ell^2 \cdot 2^{\OO(k\log k)} n^{\OO(1)})$ time, and this completes the proof of Theorem~\ref{thm:fvs}. 
%by Lemma \ref{lem:annotated}.
%As $t \leq k$, this completes the proof of Theorem~\ref{thm:fvs}.

\section{(1+$\epsilon$)-factor parameterized by feedback vertex set}\label{sec-ptas}

%In this section, we design an $ (1+\epsilon)$-factor approximation algorithm for \lesp~ problem parameterized by size of a feedback vertex set, where $ 0 < \epsilon <1 $.
% In particular, we prove the following theorem.

%\begin{theorem}[$\star$]\label{thm:fvsapprox}
%	Given a graph $G$ a set $S\subseteq V(G)$ of size $k$, an integer $\ell$ such that $G-S$ is forest. There exists an algorithm for \esp, which runs in time $\OO(2^{\OO(k\log k)}n^{\OO(1)})$ time, and returns a shortest path $P$ such for vertex $v\in V(G)$, $d_G(u,v)\leq (1+\epsilon)\ell$, where $n$ is the number of vertices of the graph $G$. 
%\end{theorem}

%In this section, we design a  $ (1+\epsilon)$-factor FPT approximation algorithm for the \mespshort problem parameterized by the feedback vertex set ($ fvs $).  
\begin{theorem}\label{thm:fvsapprox}
	For any $\epsilon>0 $, there is an $ (1+\epsilon)$-factor approximation algorithm for \mespfappr~running in time $\OO(2^{\OO(k\log k)}n^{\OO(1)})$.
	
	%	Given a graph $G$ and a set $S\subseteq V(G)$ of size $k$ such that $G-S$ is a set of disjoint paths, the \mespdpd~problem can be solved in   $\OO(2^{\OO(k\log k)}n^{\OO(1)})$ time, where $n$ is the number of vertices of the graph $G$. 
\end{theorem}

We make use of our algorithm in Theorem \ref{thm:fvs} 
%\cref{sec:fvs}~of \mespfvs (\cref{thm:fvs}) 
that runs in  $\OO(2^{\OO(k\log k)}\ell^kn^{\OO(1)})$ time. 
Notice that, 
%the eccentricity parameter $ \ell $ occurs in the running time due to the number of skeletons 
%that contains term 
$ \ell ^k $  comes because of the number of skeletons 
(Observation \ref{lem:skelnum}). 
Specifically, 
for the 
%This $ \ell^k $ term comes from the 
function $f:X\rightarrow [\ell]$ that maintains a distance profile of the set of vertices of $S$ that do not appear on $P$. 
%For the case of $ fvs $, 
To design a $ (1+\epsilon)$-factor FPT approximation algorithm, 
%our main idea is to 
we replace the image set $ [\ell] $ with a set of fixed size using $ \epsilon $ such that we  approximate the shortest distance of the set of vertices of $S$ that do not appear on $P$, with the factor $ (1+\epsilon) $. The rest is similar to 
%of the algorithm is similar to 
Theorem \ref{thm:fvs}. Below we describe the  result, the procedure and its correctness in detail. 

\noindent Let the function $f:X\rightarrow \{ \epsilon\ell, \ell\}$ 
%. Here the function $f$ 
denote the approximate shortest distance of each vertex $x\in X$ from a hypothetical solution $P$ of \mespfappr. Formally, 
%In particular, we define as follows:

\medskip

\begin{equation*}
 f(v) = 
\begin{cases}
\epsilon\ell  & \text{if}~ d_G(v, P) < \epsilon\ell,\\
\ell  & \text{if}~   \epsilon\ell \leq d_G(v, P)  \leq \ell.
\end{cases}
\end{equation*}

%$ f(v)= \epsilon\ell $ means $ d_G(v, P) < \epsilon\ell $ and 
%$ f(v)= \ell $ means $ d_G(v, P) \geq \epsilon\ell $.

% we guess that whether for a vertex $v\in X$, distance from $v$ to the hypothetical solution path is $<\epsilon \ell $, or $\geq \epsilon \ell$. For the first former case we assign $f(v)=\epsilon \ell$, and for the later case we let $f(v)=\ell$. The remaining procedure is same as the algorithm for \mespfvs in  \cref{thm:fvs}~of \cref{sec:fvs}. Note that here we have just $2^k$ choices for $f$  which gives us the running time $\OO(2^{\OO(k\log k)}n^{\OO(1)})$. 

%\noindent We let the function $f:X\rightarrow \{ \epsilon\ell, \ell\}$. In particular,
% we guess that whether for a vertex $v\in X$, distance from $v$ to the hypothetical solution path is $<\epsilon \ell $, or $\geq \epsilon \ell$. For the first former case we assign $f(v)=\epsilon \ell$, and for the later case we let $f(v)=\ell$. The remaining procedure is same as the algorithm for \mespfvs in  \cref{thm:fvs}~of \cref{sec:fvs}. Note that here we have just $2^k$ choices for $f$  which gives us the running time $\OO(2^{\OO(k\log k)}n^{\OO(1)})$. 

\medskip

\noindent \textbf{Correctness.}
%Suppose that $P$ is a path returned by the algorithm in  \cref{thm:fvs}, when the function $f$ is defined as $f:X\rightarrow \{\epsilon\ell, \ell\}$.
Suppose that $P^*$ is a shortest path, with eccentricity $\ell$ and the function $f$ as defined in the proof of  Theorem \ref{thm:fvs}, 
returned by the algorithm in  Theorem \ref{thm:fvs}.   %\todo[]{Change the four line above}
We prove that for each vertex $v\in V(G)$, $d_G(v,P)\leq (1+\epsilon)\ell$. 
Observe that for a vertex $x\in X$, 
if $1 \leq d_G(x,P^*)<\epsilon \ell$, then for a correct guess of $f$,  $f(x)=\epsilon \ell$ and $ d_G(x,P^*) < \epsilon\ell$. Also if $\epsilon\ell \leq d_G(x,P^*)\leq \ell$, then for a correct guess of $f$, $f(x)=\ell$ and $ d_G(x,P^*)\leq \ell$. 
Recall that,
%in the description of the algorithm 
in the algorithm when we construct instances for a good function $\gamma$ (reducing to instance of 
%In step while doing Reduction to instances of 
\amesp), 
%of the algorithm in ~\cref{sec:disjoint})\ 
we remove such vertices to construct an instance of \amesp. 
The assumption (or guess) we made was that the eccentricity requirement for $v$ is satisfied using $x$. 
More explicitly, we use the following conditions: if $f(x)= \epsilon \ell ~(\text{resp,}~f(x)=\ell)$, then the eccentricity requirement for the vertex $v$ is satisfied using $x$ if $d_G(v,x) \leq \ell ~(\text{resp,}~ d_G(v,x) \leq \epsilon\ell)$. Now consider a vertex $v\in V(G)\setminus S$. Suppose that there exists a  shortest path from $v$ to $P^*$ containing no vertex from $S$, then by the description and correctness of algorithm of  Theorem \ref{thm:fvs}, we obtain that $d_G(v,P)\leq \ell$. Next, suppose that the 
shortest path from $v$ to $P^*$ contains a vertex $x\in X$, then  $d_G(v,x)+d_G(x,P^*)\leq \ell$. Therefore, for such vertices, while $d_G(x, v)\leq \ell$ and $d_G(x,P)<\epsilon\ell$, we obtain that  $d_G(v,P)\leq d_G(x,v)+d_G(x,P)\leq \ell+\epsilon \ell= (1+\epsilon) \ell$ and similarly, if $d_G(x,v)\leq \epsilon \ell$ and $d_G(x,P) \leq \ell$, then $d_G(v,P)\leq d_G(x,v)+d_G(x,P)\leq \epsilon\ell+ \ell = (1+\epsilon) \ell$. This completes the correctness of the proof of Theorem \ref{thm:fvsapprox}.

\section{Disjoint Paths Deletion Set}\label{sec:paths}
\vspace{-2mm}

%To eliminate the eccentricity parameter from the running time, we construct a set $Q$ (in Lemma \ref{lem:disjointpath}) of possible distance values of  disjoint paths deletion set $S$ to a solution path such that $|Q|$ is bounded by a function of $|S|$. % which    helps us to obtain an FPT algorithm with    ${\sf dpd}$ only. 

%\section{Proof from Section \ref{sec:paths}}\label{app:paths}
%\vspace{-2mm}

In this section, we design an FPT algorithm for the \mespshort problem parameterized by the disjoint paths deletion set ($dpd$).

\begin{theorem}\label{thm:disjoint}
There is an algorithm for \mespdpd running in time $\OO(2^{\OO(k\log k)}n^{\OO(1)})$. 

%	Given a graph $G$ and a set $S\subseteq V(G)$ of size $k$ such that $G-S$ is a set of disjoint paths, the \mespdpd~problem can be solved in   $\OO(2^{\OO(k\log k)}n^{\OO(1)})$ time, where $n$ is the number of vertices of the graph $G$. 
\end{theorem}

We make use of the algorithm 
%in Theorem \ref{thm:fvs} 
%our algorithm in  \cref{sec:fvs}~
for \mespfvs (Theorem \ref{thm:fvs} ) 
that runs in  $\OO(2^{\OO(k\log k)}\ell^kn^{\OO(1)})$ time. Notice that, the eccentricity parameter $ \ell $ occurs in the running time due to the size of skeletons that contain $ \ell ^k $ (\cref{lem:skelnum}) term. Now this $ \ell^k $ term comes because of the function $f:X\rightarrow [\ell]$ that is defined to maintain a distance profile of the set of vertices of $S$ that do not appear on $P$. For the case of $ dpd $, we can show that there is a set $Q \subseteq [\ell]$ with  $|Q| \leq 2k^2$ (\cref{lem:disjointpath}~which we prove below). We can define a function $f:X\rightarrow Q$ that will maintain the distance profile of the set of vertices of $S$ that do not appear on $P$. The rest of the algorithm is exactly as similar as for \cref{thm:fvs}. Hence we obtain the following result.

\begin{lemma}\label{lem:disjointpath}
	Let $ (G, S, k) $ be a yes instance of \mespdpd, and $ P $  be a hypothetical solution. Then there is a set $ Q \subseteq [\ell]$ of size $\leq  2k^2$  such that for each $w\in S$, $d_G(w, P) \in Q$. Moreover, one can construct such a $Q$ in $O(k^2 n^2) $ time.
	
\end{lemma}

\begin{proof}
	Let $ Q $ be the set defined as $ Q= \{ d_G(x, y),  d_G(x, y)+1, d_G(x, y)-1 \colon x,y \in S \} $. Clearly, $ |Q| \leq 2k^2 $ and  the set $ Q $ can be computed in $ O(k^2n^2) $ time. Now it remains to show that for any    solution path $P $ for a yes instance $(G,S,k)$ of \mespdpd problem, we have that for every $ w\in S $, $d_G(w, P) \in Q$. Firstly, observe that if $ w \in S \cap P $, then $ d_G(w, P) =0 $ and clearly the value $0 \in Q  $, as $ d(x,x)=0 $ for any $ x \in S $. Now  for each $ w \in S \setminus P $, let $ P_w $ be a shortest path from 	$ w $ to $ P $ and let $w^* = P\cap P_w$ and $z $ be the nearest vertex in $ S $ to $ w^* $ on $ P_w $. Clearly, $ z \in S \cap P_w $ and $d_G(w, P)=  d_G(w, z)+ d_G(z, w^*)$. Let $ M= P \cap S $. Now consider the following three cases. 
(the cases are illustrated in Figure \ref{fig:disjoint_path}) 	

\noindent\textbf{Case 1.} If $ w^* \in S  $, then $ z=w^* $. And  the value $ d_G(w, P)$ is essentially $ d_G(w, w^*)  $. As $ w, w^* \in S $ so $d_G(w, w^*) \in  \{ d_G(x, y) \colon x,y \in S \} \subseteq Q$. 

\noindent\textbf{Case 2.} If $ w^* \notin S  $ and $w^*\notin N(M)$, then there must be  a subpath $ P_{j} $ of some pair $(m_j,m_{j+1})$ such that $ w^* \in P_{j} $. Since $ w^* $ is not an  endpoint of the subpath $ P_j $, we have that $ d_G(z, w^*)=1 $ and $ d_G(w,w^*)= d_G(w, z) + 1 $. As $ w, z \in S $, so  $d_G(w, w^*) \in  \{ 1+ d_G(x, y) \colon x,y \in S \} $ which is a subset of $ Q $.  

\noindent\textbf{Case 3.} If $ w^* \notin S \cap P  $ and $w^*\in N(M)$, then there must be  a subpath $P_{j} $ of some pair $(m_j,m_{j+1})$ such that $ w^* \in P_{j} $. Since $ w^* $ is an  endpoint of the subpath $ P_j $, we have that the value $ d_G(w, w^*) $ must be either $ d_G(w, m_j)-1 $ or  $ d_G(w, m_{j+1})-1 $. As  $w,  m_j , m_{j+1} \in S $, so $d_G(w, w^*) \in  \{  d_G(x, y)-1  \colon x,y \in S \}$ which is a subset of $ Q $.	
\qed
\end{proof}

%\end{proof}	
%

\begin{figure}[ht!]
	\centering
	\includegraphics[width=0.8\textwidth]{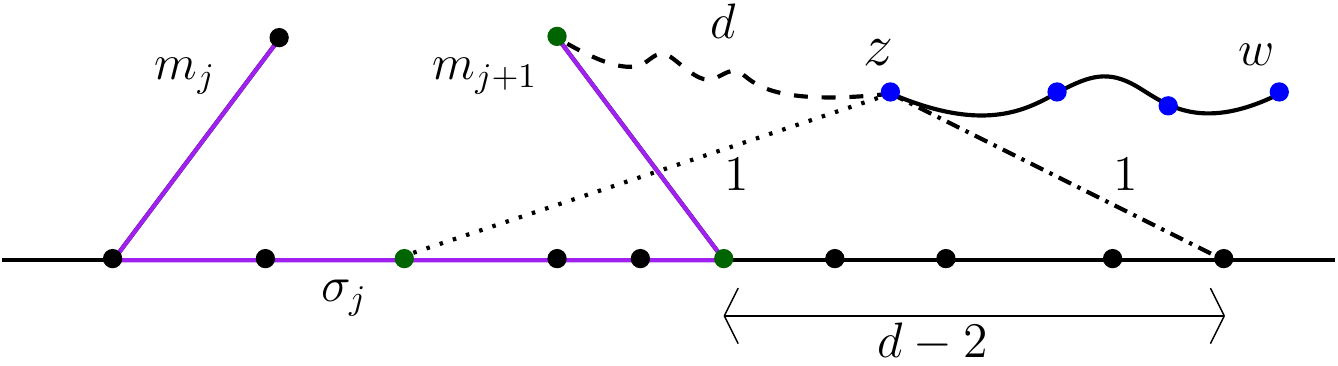}
	\caption{Illustration of the proof of \cref{lem:disjointpath}. The dashed line, the dotted line and the dashed dotted line describes the Cases  1, 2 and 3 respectively. The green colored vertices indicate the vertex $w'$ in the respective cases. The blue colored vertices represent the vertex set $Z$. The subpath between $m_j$ and $m_{j+1}$  is represented $P_j$. }
	\label{fig:disjoint_path}
\end{figure}

		\section{Split Vertex Deletion Set}\label{sec:split}

%A graph $G$ is called as \emph{split graph} if the vertices of $G$ can be partitioned into two sets $ C $ and $ I $ such that $G[I]$ is an independent set and $G[C]$ is  a clique. A set $ S\subseteq V(G)$ is called a  \emph{split vertex deletion set} ($svd$) of $ G $ if $ G-S$ is a split graph. 
In this section, we design an FPT algorithm for the \mespshort problem parameterized by the split vertex deletion set ($svd$). Let $(G,S, k,\ell)$ be a yes instance of \mespsvd, and $P$ be a hypothetical solution where $S$ is a split vertex deletion set. Our main objective is to get $P$. Towards that, our main idea is to partition the vertices from $ S $ that does not appear in $ P $ into a constant number of parts such that we have an assurance that any solution (if it exists) respecting the partition gives us a solution. In particular,  we partition the vertices from $ S $ that does not appears in $ P $, say $ X $ into five disjoint sets $ X_{=1},  X_{=2},$ $ X_{=3},  X_{=4}$ and $ X_{\geq 5} $ such that $ X_{=i}, i\in [4]$ is the set of vertices of $X$ that are at a distance exactly $ i $ from the hypothetical solution path, and  $  X_{\geq 5} $ is the set of vertices that are at a distance at least five from the hypothetical solution path. Rest of the process we describe below explicitly. The following is the main result of this section. 

\begin{theorem}\label{thm:split}
	There is an algorithm for \mespsvd running in time $ O(2^{\OO(k\log k)}\cdot n^{O(1)}) $.

	%Given a graph $G$ and a set $S\subseteq V(G)$ of size $k$ such that $G-S$ is a split graph,  \mespsvd can be solved  in time $ O(2^{\OO(k\log k)}\cdot n^{O(1)}) $, where $n$ is the number of vertices in $G$. 
\end{theorem}
\begin{proof}
Consider an instance $(G,S,k)$ of \mespsvd.	Let $C\uplus I$ be partition of $V(G-S)$  such that $C$  and $I$ induce a clique and independent set, respectively. First, we describe our algorithm. We assume that $G$ is a connected graph, else it is a trivial no instance. The algorithm first guesses a set $M\subseteq S$ which is an (exact) intersection of $S$ with the vertex set of the hypothetical solution. Observe that $|V(P)\cap C|\leq 2$, as every shortest path in the graph is also an induced path in the graph. The algorithm guesses a set $C^*\subseteq V(C)$ of size at most 2 which is an (exact) intersection of $C$ with the vertex set of the hypothetical solution. Next, the algorithm uses  \cref{lem:permutation}~to find the (unique) ordering $\pi$ of the vertices in $M\cup C^*$ on the hypothetical solution path. Let $|M\cup C^*|=t$ and $\{m_1,m_2, \dots, m_{t}\}$ be the guessed ordering $\pi$ of vertices of $M\cup C^*$ in the hypothetical solution path and let $X= S\setminus M$. Next, we guess endpoints of the solution path. Let $m_0, m_{t+1}$ be the guessed endpoints of the hypothetical solution path such that $m_0$ is the {\em first} vertex and $m_{t+1}$ is the {\em last} vertex on the hypothetical solution path (we might also obtain $m_1,m_{t+1}$ such that $m_0=m_1$ or $m_t=m_{t+1}$, or $m_0=m_{t+1}$). Observe that for a correct guess $M,C^*,m_0,m_{t+1}$ and ordering $m_0,m_1,\dots,m_{t+1}$, following must be satisfied, else we discard the combination of ordering $M,C^*,m_0,m_{t+1}$ and make another combination of guesses.  If $m_im_j$ is an edge in $G$, then $j=i+1$. 
%(2) If a vertex $vm_i, i\in [0,t]$ is an edge in $G$, then $v$ can only be adjacent to $m_{t+1}$. 
We assume that minimum eccentricity is greater than four; otherwise, we can solve the problem in polynomial time (by \cref{lem-fixed}).
   
Next, we guess a partition ${\cal X}$ of the set $X$ in 5 disjoint sets $ X_{=1},  X_{=2},$ $ X_{=3},  X_{=4}$ and $ X_{\geq 5} $ such that $ X_{=i}, i\in [4]$ is the set of vertices of $S$ that are at a distance exactly $ i $ from the hypothetical solution path, and  $  X_{\geq 5} $ is the set of vertices that are at a distance at least five from the hypothetical solution path. For a vertex $ v \in X_{=i}, i\in[4]$ (resp. $X_{\geq 5}$), if $ d_G(v,m_j)< i  $ (resp. $d_G(x,m_j)< 5$), for some $j\in [0,t+1]$ then we discard such partition and guess another partition of $ X $. 
%For a vertex $ v \in X_{=i} $ (resp. $X_{\geq 5}$), if $ d_G(v,m_j)=i  $ (resp. $d_G(x,m_j)\leq \ell$), for some $j\in [0,t+1]$ then we say that $ v $ is  covered by a vertex in $m_i\in M$.  
Let $ X' =  X_{=1} \cup X_{=2} \cup X_{=3} \cup X_{=4} \subseteq X $ be the set of all vertices in $X$ such that for  each $v\in X'$,  $ d_G(v,m_j)>i,$ where $ v \in X_{=i} $ where $ i\in [4]$ for all $j\in [0,t+1]$. In this case,  for each vertex in $X'$ there should exist a vertex in $ m_j$ to $ m_{j+1}$ subpath in the hypothetical solution that comes from $G[(V(G)\setminus S)\cup \{m_j,m_{j+1}\} ]$, for some $j\in [0,t+1]$. Recall  that for every $j\in [0,t]$, if $m_jm_{j+1}$ is an edge in $G$ then $m_jm_{j+1}$ is an edge in the hypothetical solution also.
  
    Next, we guess a function $g:X'\rightarrow [0,t]$. Where for a vertex  $ v\in X' $, $g(v)=j$ represents the guessed pair $(m_j, m_{j+1})$ such that there exists a vertex $u \in I$ in $m_j$ to $m_{j+1}$ subpath such that $d_G(u,v)=i$, if $v\in X_{=i}$. If a function $g$  satisfies the following, then we discard that guess of $g$. For some $j\in [0,t]$ and $v\in g^{-1}(j)$, there does not exist any vertex $u\in I$ satisfying (i) $u$ is adjacent to both $m_{j}$ and $m_{j+1}$ where  $(m_{j},m_{j+1}) \notin E$, (ii) $ d(u,v)= i $, where $ v \in X_{=i} $ for some $ i \in [4] $. For each $M\subseteq S$, $X=S\setminus M$, $C^*\subseteq C$, $m_0,m_{|M\cup C^*|+1}\in V(G)$, a partition ${\cal X}= X_{=1}\uplus X_{=2}\uplus X_{=3}\uplus X_{=4}\uplus  X_{\geq 5} $ of $X$, $X'$ as defined above, and a function $g:X'\rightarrow [0,t]$, which are as described above and not discarded by the above description we proceed as follows: For each $i$ such that $m_im_{i+1}$ is not an edge in $G$, we arbitrarily pick a vertex $u\in I$ such that $u$ is adjacent to both $m_{j}$ and $m_{j+1}$ and $ d_G(v,u)=i $ whenever $ v \in X_{=i} $ and $v \in g^{-1}(j)$ together holds. This completes the description of the algorithm. We output the path $P$ by concatenating $m_im_{i+1}$ edge or $m_jum_{j+1}$ subpaths for every $j\in [0,t]$. 
%    \end{proof}
%    \todo{define concatenate in prelim }

\medskip 

\noindent{\textbf{Running time:}} Next, we analyze the running time of our algorithm. The number of subsets of $S$ is $2^k$, which gives us $2^k$ choices for $M$. The number of choices for $C^*,m_0,m_{|M|+1}$ is at most $\OO(n^4)$. To find the permutation (ordering) of $M\cup C^*$ we have only one choice, which can be found in $\OO(n^{\OO(1)})$ time by  \cref{lem:permutation}. The choices for partition $\cal X$ are bounded by $\OO(5^k)$. The choices for the function $g$ are bounded by $\OO(k^k)$. Therefore, we obtain running time in  $\OO(2^{\OO(k\log k)} \dot n^{\OO(1)})$ of our algorithm.

\medskip
   
   \noindent  {\textbf{Correctness:}} Next, we prove the correctness of our algorithm. We show that a shortest path $P$ is a solution to $(G,S,k)$ of \mespsvd if and only if there exists a $M\subseteq S$, $X=S\setminus M$, $C^*\subseteq C$, $m_0,m_{|M\cup C^*|+1}\in V(G)$, a partition ${\cal X}= X_{=1}\uplus X_{=2}\uplus X_{=3}\uplus X_{=4}\uplus X_{\geq 5} $ of $X$, $X'$ as defined above, and a function $g$, such that $P$ satisfies the properties of  $M,X, C^*, m_0, m_{|M\cup C^*|+1},{\cal X},X',g$.
    
     In the forward direction suppose that $P$ is a solution to $(G,S,k)$ of \mespsvd. Let $M\subseteq V(P)\cap S$, $X=S\setminus M$, $C^*=V(P)\cap C$. Let $\pi$ be the ordering of vertices in $M\cup C^*$ on $P$. Let $m_0,m_{|M\cup C^*|+1}$ be the first and last endpoints of $P$, respectively. Let ${\cal X}= X_{=1}\uplus X_{=2}\uplus X_{=3}\uplus X_{=4}\uplus  X_{\geq 5} $ be partition of $X$ according to distances of vertices from $P$. Let $ X' \subseteq X $ be the set of all vertices in $X$ such that for  each $v\in X'$,  $ d_G(v,m_j)>i, i\in [4]$ for all $j\in [0,t+1]$. For defining function $g$, if for a vertex $v\in X'$ there exists $u_i$ in $m_j$ to $m_{j+1}$ path, then set $g(v)=j$. Observe that $P$  satisfies properties of $M,X, C^*, m_0, m_{|M\cup C^*|+1},{\cal X},X',g$.
     
     In the backward direction suppose that there exists a $M\subseteq S$, $X=S\setminus M$, $C^*\subseteq C$, $m_0,m_{|M\cup C^*|+1}\in V(G)$, a partition ${\cal X}= X_{=1}\uplus X_{=2}\uplus X_{=3}\uplus X_{=4}\uplus X_{\geq 5} $ of $X$, $X'$ as defined above, and a function $g$, such that $P$ is a shortest path that satisfies the properties of  $M,X, C^*, m_0, m_{|M\cup C^*|+1},{\cal X},X',g$. We claim that $P$ is a solution to $(G,S,k)$ of \mespsvd. Towards a contradiction, suppose that  $P$ is not a solution $(G,S,k)$ of \mespsvd and there exists a shortest path $P_1$ which is a solution to $(G,S,k)$ of \mespsvd such that $\ell$ is the minimum integer such that for all $u\in V(G)$, $d_G(u, P_1)\leq \ell$. Let $v\in V(G)$ such that $d_G(v, P_1)\leq \ell$ and $d_G(v, P)>\ell$. Observe that for a correct guess of ${\cal X}$, for every vertex $u\in X^*= \cup_{i=1}^4 X_{=i}$, distance of $u$ to $P$ is optimal, that is for any solution $P^*$ of $(G,S,k)$, $d_G(u,P^*)=d_G(u,P)$.  Therefore, $v\notin X^*$. Recall that $\ell\geq 5$, otherwise the algorithm would have solved the instance completely using  \cref{lem-fixed}. Let $ P_1^v $  be a shortest path from $ v $ to $ P_1 $. Observe, that since $G-S$ is a split graph, every induced path is of length at most $3$. As $P^v_1$ is a shortest path of length at least $5$ ($\ell\geq 6$), then by induced property of shortest path, we have that that there is a vertex $u\in S$ on $P^1_v$ such that $d_G(u,P)=j\leq 4$.  As ${\cal X}$ is a correct guess so for each $u\in X_{=i}, i\in [4]$, we have that $d_G(u, P)=i $. By the observation that $d_G(u, v)\leq \ell -i$ (along the path $ P_1^v $), we obtain that  $d_G(v, P)\leq \ell$, a contradiction. Hence $P$ is also a solution to $(G,S,k)$. This completes the proof of correctness of our algorithm and hence proof of  \cref{thm:split}. \qed 
     \end{proof}

\section{W[2]-hardness for \mespcvd}\label{sec:hardness}

\vspace{-2mm}

In this section, we show that \mespcvd is \wh. This construction is based on \cite{dragan2017minimum} where the authors prove that \mespecc is \wh.
%, we call the problem as \mespcvd. In \mespcvd,  we are given an undirected graph $G$, a set $S\subseteq V(G)$ of size $k$ such that $G-S$ is a chordal graph, the parameter is $k$, and the aim is to check whether there exists a shortest path $P$ in $ G $ of eccentricity at most $\ell$.

\begin{theorem}\label{thm:hardness}
\mespcvd is $W[2]$-hard.
\end{theorem}

%\begin{theorem}\label{thm:hardness}
% The problem \mespcvd is $W[2]$-hard.
%\end{theorem}
%
%
%\defparproblem{\mespcvd}{ An undirected graph $G$, a set $S\subseteq V(G)$ of size $k$ such that $G-S$ is a chordal graph.}{$k$}{Find a shortest path $P$ in $ G $ of minimum eccentricity.}

\begin{proof}

Towards proving that the \mespcvd is \wh, we give a polynomial time parameter preserving reduction from \dsfull (\ds in short) parameterized by solution size to \mespcvd. In \ds, we are given a graph $G$, and an integer $k$, the aim is to decide whether there exists a set $S\subseteq V(G)$ of size at most $k$ such that for every vertex $v\in V(G)$, $N[v]\cap S\neq \emptyset$. It is well known that \ds is \wh~parameterized by the solution size $k$ \cite{downey2012parameterized}.

%\begin{proof}
 Consider an instance $(G,k)$ of \ds. First, we describe construction of an instance $(H,S,k')$ of \mespcvd. 
 Let $V(G)=\{v_1,\dots,v_n\}$. To construct the graph $H$ we apply the following procedure (see \cref{fig:hardness} for an illustration of the construction). 
\begin{enumerate}
    \item Initialize $V(H)=V(G)$, that means, we add a vertex $v_i$ in $V(H)$ corresponding to each vertex $v_i$ in $G$. For each pair $v_i,v_j\in V(H), i\neq j$, we add $v_iv_j$ edge in $H$. Let $V^*$ be the set of vertices added in this step to $V(H)$. 

\item Add $k$ sets $U_1, U_2, \ldots, U_k$, where $ U_i = \{ u_{i1}, u_{i2} , \ldots , u_{in} \} $ to $V(H)$. Each $U_i, i\in [k]$ is a set of $n$ vertices. The vertex $u_{ij}$ in each set $U_i$ where $i\in [k],j\in [n]$, represents $i$th representative of $j$th vertex in $G$ and we add $k$ representatives for each vertex. 

\item Add a set of $ (k-1) $ vertices $ \{z_1, z_2, \ldots, z_{k-1} \} $ to $V(H)$. For each $ i\in [k-1]$, make all vertices in $U_i \cup U_{i+1}$ adjacent to $ z_i $, that is, add the set $\{u_{ij}z_i,u_{(i+1)j}z_i \colon i\in [k],\  j\in[n]\}$ of edges to $E(H)$. 

\item Add four	vertices, $ s $, $ a $, $ b $, and $ t $ to $V(H)$, also add two sets $A=\{a_1,a_2,\dots, a_{k}\}$ and $B=\{b_1,b_2,\dots, b_{k}\}$ each of $k$ vertices to $V(H)$. Add edge sets $\tilde{A}=\{a_ia_{i+1}\colon i\in [k-1]\}\cup \{sa_1,a_{k}a\}$ and $\tilde{B}=\{b_ib_{i+1}\colon i\in [k-1]\}\cup \{tb_1,b_{k}b\}$ to $V(H)$. Observe that we have added a path from $s$ to $a$ of length $k+1$ using set $A$ of vertices and edges $\tilde{A}$. Similarly we have added a path from $b$ to $t$ of length $k+1$ using set $B$ of vertices and edges $\tilde{B}$.

\item Make the vertex $  a $ adjacent to all vertices in $ U_1 $ and make the vertex $ b $ adjacent to all vertices in $ U_k $. That is add the set  $\{u_{1j}a,u_{kj}b\colon j\in [n]\}$ of edges to $E(H)$. 

\item For each vertex $u_{ij} \in U_i, i\in [k], j\in [n]$ and each vertex $v_{t}\in N_G [ v_j ] $ add a set $W^{ij}_{v_t}=\{w^{ij}_{t1},w^{ij}_{t2},\dots, w^{ij}_{t(k)}\}$ of $k$ vertices to $V(H)$. Add edge set $\ty^{ij}_{v_t}=\{w^{ij}_{tp}w^{ij}_{t(p+1)}\colon p\in [k-1]\}\cup \{u_{ij}w^{ij}_{t1},w^{ij}_{t(k)}v_t\}$ of edges to $E(H)$. Observe that we have added a $u_{ij}$ to $v_t$ path of length $k+1$ using vertex set $W^{ij}_{v_t}$ and edge set $ \ty^{ij}_{v_t}$ , for each $i\in [k], j\in [n], v_t\in N_G [ v_j ] $. Intuitively if $v_jv_t$ is an edge in the graph $G$ (or $t=j$, as we consider closed neighbourhood of $v_j$), then we add a path of length $k+1$ between each vertex corresponding to $v_j$ that we have added in sets $U_i, i\in [k]$ in step 2, and vertex corresponding to vertex $v_t\in V^*$ that we have added in step 1. 

\item (Make a ladder for $i$th representative of $v_j$ vertex  with the vertices corresponding to $N[v_j]$ in $V^*$.) Next, we take the $u_{ij}$ to $v_t$ paths added in the above step, $u_{ij} \in U_i, i\in [k], j\in [n]$ and $v_{t}\in N_G [ v_j ] $, we further add more edges to vertices in these paths and construct a \emph{ladder} structure. In particular we do as follows.  For each vertex $u_{ij} \in U_j, i\in [k], j\in [n]$ and each pair of vertices $v_{t}v_{t^*}\in N_G [ v_j ] $ we take the  $u_{ij}$ to $v_t$ path and $u_{ij}$ to $v_{t^*}$ path added in above step. We have $W^{ij}_{v_t}=\{w^{ij}_{t1},w^{ij}_{t2},\dots, w^{ij}_{tk}\}$ vertex set corresponding to  $u_{ij}$ to $v_t$ path and $W^{ij}_{v_{t^*}}=\{w^{ij}_{t^*1},w^{ij}_{t^*2},\dots, w^{ij}_{t^*k}\}$ vertex set corresponding to  $u_{ij}$ to $v_t^*$ path in the above step. We further add edges  $\hat{W}^{ij}_{v_tv_{t^*}}=\bigcup_{p\in [k]}w^{ij}_{tp}w^{ij}_{t^*(p)}\cup \{w^{ij}_{tq}w^{ij}_{t^*(q+1)}\colon q\in [k-1]\}\cup \{w^{ij}_{tk}v_{t^*}\}$ to $E(H)$. We say that the graph $H$ induced on vertex sets $W^{ij}_{v_t}\cup W^{ij}_{v_{t^*}}\cup \{u_{ij},v_t,v_{t^*}\}$ together with edge set $\ty^{ij}_{v_t}\cup\ty^{ij}_{v_{t^*}}\cup \{v_tv_{t^*}\}$ is a {\em ladder} between $u_{ij}$ and $v_t,v_{t^*}$.

\end{enumerate}

	This completes the description of the graph $H$.  Observe that the number of vertices in $ H $ is polynomial in $ n $ and $ k $, and the construction can be done in polynomial time. Let $ X= \{ z_1, z_2, \ldots z_{k-1}\} \cup \{a, b\} $. Clearly $ |X|=k+1 $. Let $k'=k+1$. This completes the description of the construction of $(H,k')$. 
	We claim that $G $ has a dominating set of size at most $ k $ if and only if $ H $ has a shortest path $ P $ with eccentricity at most $ (k+1)$. 
	%The correctness of the above claim is presented in \cref{sec:hardnessapp}.
%	Due to space constraints, 
%	the figure illustrating the construction (\cref{fig:hardness}) 
%	and the correctness of the above claim is presented in \cref{sec:hardnessapp}. 
	%and the correctness of our construction. See  ~for an illustration of the construction. The correctness of our construction can be found in \cref{sec:hardnessapp}.
\end{proof}

\noindent{\bf Correctness. }
We show that  $(H,S,k')$ is a valid instance of \mespcvd.
First, we show that $X$ is a chordal deletion set of $H$.

\begin{claim}
	$H-X$ is a chordal graph.
\end{claim}
\begin{proof}
	Let $ C $ be a cycle of length at least four in $  H \setminus X$. As $ \{v_1, v_2, \ldots v_n\} $ forms a clique in $ G $ so $ C $ can have at most two vertices from  $ V^* $. Now in each component of $ H \setminus (X \cup V^*) $ has at most one vertex from $ U = \cup_{i=1}^k U_i $. As $ V^* $ induces a clique, so there is no induced cycle having more than one vertices from $ U $. Let $ u_i^j \in U_j$ be an arbitrary vertex in $ C $, by our construction neighbors of $u_{ij}$ are vertices of type $w^{ij}_{t1}$, which is first vertex in $u_{ij}$ to $v_t\in N_G[v_j]$ path (step 5) or of type  $w^{i'j}_{t1}$, which is first vertex in $u_{i'j}$ to $v_t\in N_G[v_j]$ path (step 7). Observe that by our construction of ladders in step 6 and 7, every pair of vertices in neighbourhood of $u_i^j$ in $H$ is adjacent (such edges are added in $E(H)$ by $\hat{Y}$ type sets). 	Now let,  $ w^{ij}_{tt'} $ be one of the nearest neighbour (say distance $ t' $) of $ u_i^j $ in $ C $. If we look at the neighbour of $ w^{ij}_{tt'} $ in $ C $, as $ C $ is a cycle of length at least 4 so the only choice is $ w^{ij}_{t^*t'}$ and $w^{ij}_{t(t'+1)}$.  Now $ w^{ij}_{tt'} $ is one neighbour of $ w^{ij}_{t^*t'} $ in $ C $, let $ w $ be the other neighbour of $ w^{ij}_{t^*t'} $, clearly $ d_{H}(u_i^j, w)= t'  $ or $ t'+1 $. In any of the cases, as our construction there is an edge between $ w^{ij}_{tt'} $ and $ w $. That leads to a contradiction. So, there is no induced cycle of length at least four in $  H \setminus X$.  Hence, we have that $H$ is a chordal graph.
\end{proof}

%	
%	Observe the following: (1) $H[V^*]$ is a clique and hence chordal, (2) $H[U_1\cup U_2\dots\cup U_k]$ is an independent set and hence chordal. (2) $H[A\cup B\cup \{s,t\}]$ is set of two disjoint induced paths and hence chordal. Let $ C $ be a cycle of length at least four in $  H -X$. As $ V^*=\{v_1, v_2, \ldots v_n\} $ forms a clique in $H$, $ C $ cannot contain more than two vertices from  $ V^* $, which are consecutive in $C$. Note that any cycle $C$ containing at least two vertices $v_t,v_{t^*}$ from $ V^*$ also contain at least two vertices from the paths $u_{ij}$ to $w$ and $u_i^j$ to $\htw$ for some $u_i^j, i\in [n],j\in [k]$, which contains a chord by the ladder construction in step 7. Note that there is no induced cycle of length more than three which contains any vertex from $ U_j, j\in [k] $, by the ladder construction in step 7.
%	 
%	  Next, we focus on the subgraph $ H' $ induced by $ u_i^j, N_G[v_i], i\in[n],j\in [k] $ and all the vertices in $Y^{ij}_w$  for every $w\in N_G[v_i]$ (See \cref{fig:hardness}). 
%	Recall that there is a ladder between $u_i^j$ and every pair of vertices $w,\htw\in N_G(v_i)$. Therefore, there is no cycle $ C $ of length at least 4 in $ H' $. Thus, there is no induced cycle of length at least four in $  H-X$.

%Let $ u_i^j $ be an arbitrary vertex in $ U_j $, by our construction there is a path of length $ k $ joining $ u_i^j $ to each vertex in $ N_G[v_i] $. 

\begin{figure}[ht!]
	\centering
	\includegraphics[width=.8\textwidth]{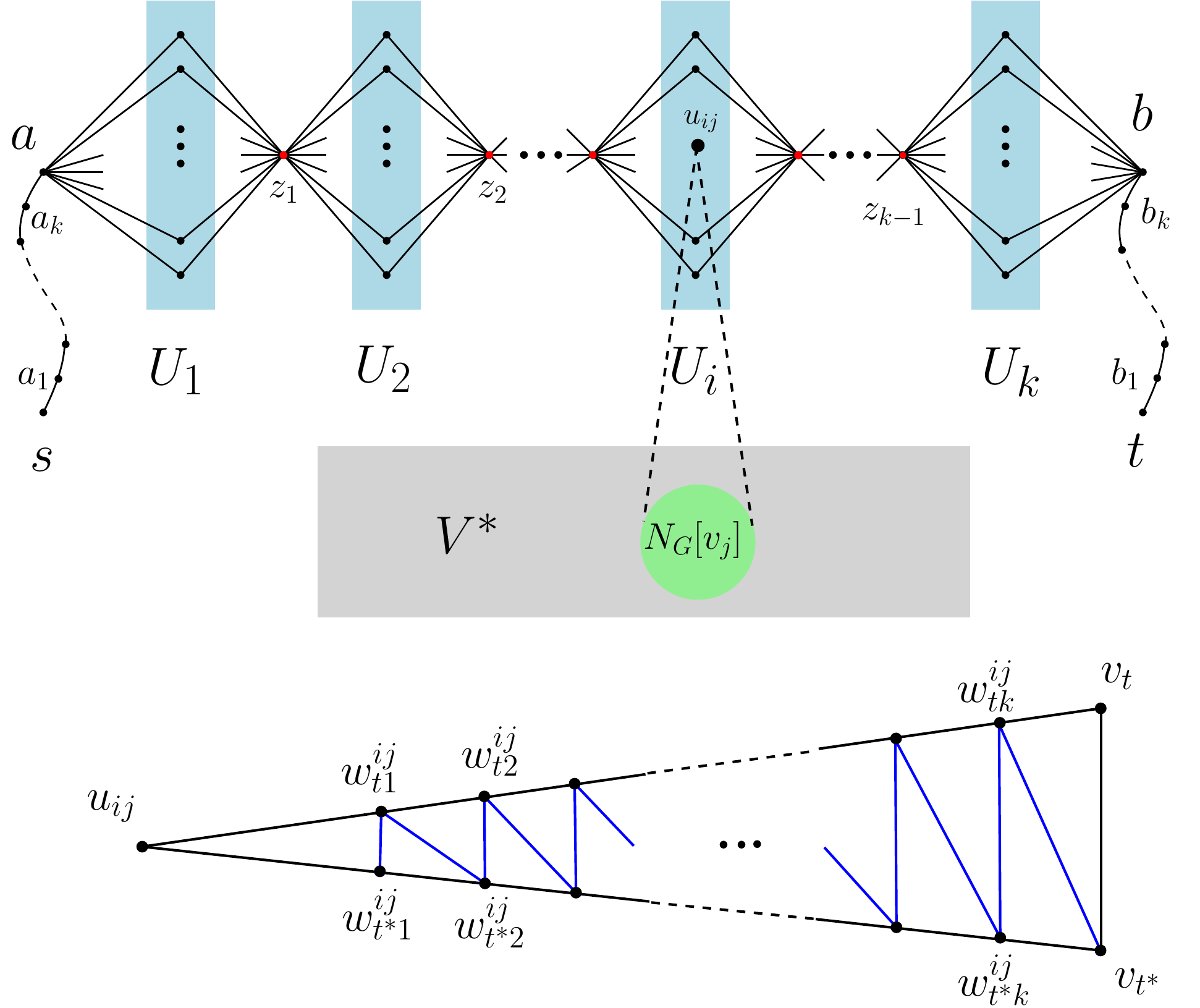}
	\caption{Reduction from \ds to \mespcvd. Illustration to the proof of  \cref{thm:hardness}.}
	\label{fig:hardness}
\end{figure}

%	
%	
%	
%	Let the parameter $\ell=k$ in the reduced instance. To complete the reduction, we prove the following claim:
%	
\begin{claim}	
	$G $ has a dominating set of size at most $ k $ if and only if $ H $ has a shortest path $ P $ with eccentricity at most $ (k+1)$.
\end{claim}

\begin{proof}
	In the forward direction, suppose that the graph $ G $ has a dominating set $ D$ of size $ k $. Let $D=\{v_{t_1}, v_{t_2},\dots, v_{t_k}\}$ and $D^*=\{a,u_{1t_1},z_1,u_{2t_2},z_2,\dots, z_{k-1},u_{kt_k}, b\}$, where $u_{it_i}$ is the $i$th representative of vertex $t_i$ in set $U_i$ added in step 2 of our construction of $H$. Observe that $P=H[D^*]$ is an induced path by our construction. Since $ D $ is a dominating set in $ G $, for each vertex $v_{t_i}\in D$, we have that each vertex $ v_{t} \in N_G[v_i] $ is at distance $ (k+1) $ from $ u_{it_i}$ in $ H $. By construction, for the remaining vertices in $ H $, there is a shortest path from $P$ of length at most $ (k+1) $.

	For the backward direction, note that each shortest path that does not contain the vertices $ a $ or $ b $ have eccentricity more than $ k+1 $, as distance from $ s $ to $ a $ (resp, $ t $ to $ b $) is exactly $ k+1 $. Moreover, as $ d_H(a,b)= 2k $ and any path between $ a $ and $ b $ passing through $ V^*$ exceeds the value $ 2k $ so each shortest path must intersects all set $ U_i$ and pass through all vertices $ z_{i'} $ where $ i\in [k] $ and $ i'\in [k-1] $. Suppose that there exists a shortest path $ P $ with eccentricity $ k+1 $ in $H$. That means for every vertex $ v \in H $ there exists a vertex $ u \in P$ such that $ d_H(u,v) \leq (k+1) $. Let $ v $ be an arbitrary vertex in $ V^* $, and there exists a vertex $ u_{ij}\in P $ such that $ d (v, u_{ij})=k+1 $. This implies, $ v \in N_G[v_j] $. As we choose $ v $ arbitrarily, so the set $ D=\{ v_j \colon u_{ij}\in V(P)\cap U_i, i\in [k]\} $ is a dominating set for $ G $ with size at most $ k $. This completes the proof.
\end{proof}
This completes the proof of \cref{thm:hardness}. \qed

	\bibliography{bibfile.bib}
	\appendix

\end{document}